\documentclass[11pt]{article}
\usepackage{color}

\usepackage{amssymb}   
\usepackage{amsthm}    
\usepackage{amsmath}   
\usepackage{mathtools}
\usepackage{stmaryrd}  
\usepackage{titletoc}  
\usepackage{mathrsfs}  
\usepackage{graphicx}
\usepackage{todonotes}
\usepackage{enumerate}
\usepackage{float}
\usepackage{multirow}
\usepackage{multicol}
\usepackage{subfigure}

\newlength{\defbaselineskip}
\newcommand{\setlinespacing}[1]%
           {\setlength{\baselineskip}{#1 \defbaselineskip}}

\theoremstyle{plain}
\newtheorem{thm}{Theorem}[section]
\newtheorem{cor}[thm]{Corollary}
\newtheorem{lem}[thm]{Lemma}
\newtheorem{prop}[thm]{Proposition}
\newtheorem{exam}[thm]{Example}

\theoremstyle{definition}
\newtheorem{defn}{Definition}[section]
\newtheorem{ass}{Assumption}[section]
\newtheorem{rmk}{Remark}[section]

\newcommand{\cB}{\mathcal{B}}
\newcommand{\cA}{\mathcal{A}}

\newcommand{\cE}{\mathcal{E}}

\newcommand{\bH}{\mathbb{H}}
\newcommand{\bP}{\mathbb{P}}
\newcommand{\bR}{\mathbb{R}}
\newcommand{\bN}{\mathbb{N}}

\newcommand{\sF}{\mathscr{F}}
\newcommand{\sP}{\mathscr{P}}

\DeclareMathOperator{\var}{var}

\makeatletter\@addtoreset{equation}{section} \makeatother

 \allowdisplaybreaks
\begin{document}

\title{ Pricing Options Under Rough Volatility with Backward SPDEs
}

\author{Christian Bayer\footnotemark[1]  \and  Jinniao Qiu\footnotemark[2]  \and Yao Yao\footnotemark[2]}
\footnotetext[1]{Weierstrass Institute for Applied Analysis and Stochastics
  (WIAS), Berlin, Germany. \textit{Email}:
  \texttt{christian.bayer@wias-berlin.de}. C.~Bayer gratefully acknowledges
  funding by the German Research Foundation (DFG) (project AA4-2 within the
  cluster of excellence MATH+).}
\footnotetext[2]{Department of Mathematics \& Statistics, University of Calgary, 2500 University Drive NW, Calgary, AB T2N 1N4, Canada. \textit{Email}: \texttt{jinniao.qiu@ucalgary.ca} (J. Qiu), \texttt{yao.yao1@ucalgary.ca} (Y. Yao). J. Qiu was partially supported by the National Science and Engineering Research Council of Canada and by the start-up funds from the University of Calgary. }

%

\maketitle

\begin{abstract}
In this paper, we study the option pricing problems for rough volatility models. As the framework is non-Markovian, the value function for a European option is not deterministic; rather, it is random and satisfies a backward stochastic partial differential equation (BSPDE). The existence and uniqueness of weak solution is proved for general nonlinear BSPDEs with unbounded random leading coefficients whose connections with certain forward-backward stochastic differential equations are derived as well.  These BSPDEs are then used to approximate American option prices. A deep leaning-based method is also investigated  for the numerical approximations to  such BSPDEs and associated non-Markovian pricing problems. Finally, the examples of rough Bergomi type are numerically computed for both European and American options.
\end{abstract}



\section{Introduction}

Let $(\Omega, {\sF}, ({\sF}_t)_{t\in[0,T]},\bP)$ be a complete filtered probability space with the filtration $({\sF}_t)_{t\in[0,T]}$ being the augmented filtration generated by  two independent Wiener processes $W$ and $B$.  Throughout this paper, we denote by $(\sF^W_t)_{t\in[0,T]}$ the augmented filtration generated by the Wiener process $W$. The predictable $\sigma$-algebras on $\Omega\times[0,T]$ corresponding to $(\sF^W_t)_{t\in[0,T]}$ and  $({\sF}_t)_{t\in[0,T]} $ are denoted by $\sP^W$ and ${\sP}$, respectively.


We consider a general stochastic volatility model given under a risk neutral
probability measure as
\begin{equation}\label{SDE-price}
\left\{
\begin{split}
dS_t&=rS_tdt+S_t \sqrt{V_t} \left( \rho \,dW_t + \sqrt{1-\rho^2} \,dB_t  \right);\\
S_0&=s_0,
\end{split}\right.
\end{equation}
where $\rho\in [-1,1]$ denotes the correlation coefficient and the constant $r$ the interest rate. We impose the following assumptions on the stochastic variance process $V$.
\begin{ass}
  \label{ass:stoch-var}
  $V$ has continuous trajectories, takes values in $\mathbb{R}_{\ge 0}$, and is
  adapted to the filtration generated by the Brownian motion $W$. We further
  assume that $V$ is integrable, i.e.,
  \begin{equation*}
    E\left[ \int_0^T V_s ds \right] < \infty, \quad T > 0.
  \end{equation*}
\end{ass}

Note that we do not assume that $V$ (or even $(S,V)$) is a Markov process or a
semi-martingale, and, in fact, our main examples will be neither.
Indeed, the motivation of this work is to extend the backward stochastic differential equation-based pricing theory
to rough volatility models. These models were put forth in
\cite{gatheral2018volatility} in order to explain the roughness of time series
of daily realized variance estimates. The idea is that the spot price process
is modeled by a stochastic volatility model, with the stochastic variance
process essentially behaving like an exponential fractional Brownian motion
with Hurst index $0 < H < 1/2$ -- in contrast to an earlier strand of
literature (see, e.g., \cite{comte2012affine}) which tried to model long
memory in the variance by fractional Brownian motion with $1/2 < H$. In the
pricing domain, rough volatility was found in~\cite{bayer2016pricing} to lead
to extremely accurate fits of SPX implied volatility surfaces with very few
parameters, in particular explaining the power law behaviour of the ATM
implied volatility skew for short maturities; see also
\cite{alos2007short,fukasawa2011asymptotic}. Since then, there have been many
new contributions to the literature of rough volatility models, including
developments of rough Heston models with closed expressions for the
characteristic functions (see \cite{el2019characteristic}), microstructural
foundations of rough volatility models (\cite{el2018microstructural}),
calibration of rough volatility models by machine learning techniques
(\cite{bayer2019deep}), a theory of affine rough Volterra processes
(\cite{jaber2019affine}) and a regularity structure (in the sense of Hairer)
for rough volatility (\cite{bayer2019regularity}), to mention just a few.

In this work, we keep the following examples specifically in mind.

\begin{exam}
  \label{ex:rBergomi}
  In the \emph{rough Bergomi model} (see~\cite{bayer2016pricing}), the
  stochastic variance is given as
  \begin{equation}
    \label{eq:rBergomi}
    V_t = \xi_t \mathcal{E}\left( \eta \widehat{W}_t \right),
  \end{equation}
  where $\xi_t$ denotes the \emph{forward variance curve} (a quantity which
  can be computed from the implied volatility surface), $\mathcal{E}$ denotes
  the \emph{Wick exponential}, i.e., $\mathcal{E}(Z) \coloneqq \exp\left( Z -
    \frac{1}{2} \var Z \right)$ for a zero-mean normal random variable $Z$,
  and $\eta \ge 0$. Finally, $\widehat{W}$ denotes a fractional Brownian
  motion (fBm) of Riemann-Liouville type with Hurst index $0 < H <
  \frac{1}{2}$, i.e.,
  \begin{equation}
    \label{eq:RL-fbm}
    \widehat{W}_t \coloneqq \int_0^t \mathcal K(t-s) dW_s, \quad \mathcal K(r) \coloneqq
    \sqrt{2H} r^{H-1/2}, \quad r > 0.
  \end{equation}
  If the correlation $\rho$ is negative, then
  Gassiat~\cite{gassiat2018martingale} showed that the discounted price $e^{-rt}S_t$ is, indeed, a
  martingale; otherwise, it may not be a martingale. But the conditions of Assumption~\ref{ass:stoch-var} are
  always satisfied.
\end{exam}

\begin{exam}
  \label{ex:rough-Heston}
  In the \emph{rough Heston model} introduced in~\cite{el2019characteristic},
  the stochastic variance satisfies the stochastic Volterra equation
  \begin{equation}
    \label{eq:rough-Heston}
    V_t = V_0 + \int_0^t \mathcal K(t-s) \lambda \left( \theta - V_s \right) ds +
    \int_0^t \mathcal K(t-s) \zeta \sqrt{V_s} dW_s,
  \end{equation}
  where the Kernel satisfies
  \begin{equation}
    \label{eq:power-law-kernel}
    \mathcal K(r) \coloneqq r^{\alpha-1}/\Gamma(\alpha), \quad r>0, \quad \frac{1}{2} <
    \alpha < 1.
  \end{equation}
  The rough Heston process also satisfies Assumption~\ref{ass:stoch-var};
  see~\cite{jaber2019affine}.
\end{exam}

For each $(t,s) \in [0,T]\times \bR^+$, denote the asset/security price process by $S_{\tau}^{t,s}$,  for $\tau\in[t,T]$, which satisfies the stochastic differential equation (SDE) in \eqref{SDE-price} but with initial time $t$ and initial state $s$ (price at time $t$).  The fair price of a European option with payoff $H$, as the smallest initial wealth required to finance an admissible (super-replicating) wealth process, is given by
\begin{align}
P_t(s):= E\left[   e^{-r(T-t)}H(S^{t,s}_T)\big| \sF_t   \right];
\end{align}
refer to \cite{cox2005local} for the cases when the discounted price $e^{-rt}S_t$ is just a local martingale. Taking $X_t= -rt+\log S_t$, we may reformulate the above pricing problem, i.e.,
\begin{align}
u_t(x):= E\left[   e^{-r(T-t)}H(e^{X_T^{t,x}+rT})\big| \sF_t  \right], \quad (t,x) \in [0,T]\times \bR, \label{European-optn}
\end{align}
subject to
\begin{equation}\label{forward-X}
\left\{
\begin{split}
dX_s^{t,x}&= \sqrt{V_s} \left( \rho \,dW_s+ \sqrt{1-\rho^2} \,dB_s  \right) -\frac{V_s}{2}\,ds, \quad 0\leq t\leq s\leq T;\\
X_t^{t,x}&=x.
\end{split}\right.
\end{equation}
Obviously, we have the relation $u_t(x)=P_t(e^{x+rt})$ a.s..

The non-Markovianity of the pair $(S,V)$ (or $(X,V)$)  makes it impossible to characterize the value function $u_t(x)$ with a conventional (deterministic) partial differential equation (PDE). Indeed, we prove that the function $u_t(x)$, for $(t,x)\in[0,T]\times \bR$, is a random field which together with another random field $\psi_t(x)$ satisfies the following backward stochastic partial differential equation (BSPDE):
\begin{equation}\label{BSPDE-value-funct}
\left\{
\begin{split}
-du_t(x)&= \Big[\frac{V_t}{2}D^2u_t(x) + \rho \sqrt{V_t}D\psi_t(x)-\frac{V_t}{2}Du_t(x)
-ru_t(x)\Big]\,dt
			-\psi_t(x)\,dW_s;\\
u_T(x)&=H(e^{x+rT}),
\end{split}\right.
\end{equation}
where the pair $(u,\psi)$ is unknown and the volatility process $(V_t)_{t\geq 0}$ is defined exogenously as in Examples \ref{ex:rBergomi} and \ref{ex:rough-Heston}.

While the BSPDEs have been extensively studied (see \cite{BenderDokuchev-2014,DuQiuTang10,Hu_Ma_Yong02,Peng_92} for instance), to the best of our knowledge, there is no available theory for the well-posednesss of BSPDE \eqref{BSPDE-value-funct} because the leading coefficient $\frac{V_t}{2}$ is neither uniformly bounded from above nor uniformly (strictly) positive from below and the terminal value $H(e^{\cdot+rT})$ may not belong to any space $L^p(\Omega\times\bR)$ for $p\in (1,\infty)$. Hence, a weak solution theory is established for the well-posedness of general nonlinear BSPDEs and associated \textit{stochastic} Feynman-Kac formula, particularly applicable to \eqref{BSPDE-value-funct}. Such nonlinear BSPDEs are further used to approximate the American option prices. Based on the stochastic Feynman-Kac formula with forward-backward stochastic differential equations (FBSDEs), we develop a deep learning-based method for numerical approximations for the solutions which are essentially defined on the (infinite dimensional) probability space due to the randomness. Accordingly, the universal approximation theorem of neural networks is generalized from finite dimensional input spaces to infinite dimensional cases in the probabilistic setting. On the basis of this approximation result, we design the schemes in the spirit of the Markovian counterpart by Hur{\'e}, Pham, and Warin \cite{hure2019some} but equipped with neural networks with changing and high input dimensions. Some numerical results are also presented for examples of  rough Bergomi type, along with an appended convergence analysis. 
Here, although the theory and application results are presented for the case of a single risky asset under rough volatility, leading to associated BSPDEs on the one-dimensional space $\bR$,  a multi-dimensional extension may be obtained under certain assumptions in a similar manner; nevertheless, we would not seek such a generality to avoid cumbersome arguments.

  Finally, let us contrast the present work with the recent work
  \cite{jacquier2019deep}. Therein, with the method developed in \cite{viens2019martingale}  the European option price in a local
  rough volatility model is expressed as a function of $t, S_t$ and an
  additional, infinite-dimensional term $\Theta$, which is closely related to
  the forward variance curve. An infinite-dimensional pricing PDE for the option price with respect
  to these variables is then formulated and solved with a discretization method
  using deep neural networks as basis functions. The focus of
  \cite{jacquier2019deep} is clearly on the mathematical finance and numerical
  side, whereas well-posedness of the path-dependent PDE is more or less
  assumed. (They do refer to \cite{ekren2014viscosity}, which, however, only
  covers the case of path-dependent PDEs with constant diffusion
  coefficients. Moreover, the arguments in \cite{jacquier2019deep} seem to
  require classical -- not viscosity -- solutions of the path-dependent PDE.) In this
  sense, our present work is complementary, as the well-posedness of the BSPDE is a serious concern of this paper. We also extend the consideration from the
  European to the American case, and provide similar type of numerical
  discretization also based on deep neural networks, but for approximation of
  the associated FBSDEs.

The rest of this paper is organized as follows. Section 2 is devoted to the well-posedness of a class of nonlinear BSPDEs and associated stochastic Feynman-Kac formula. The weak solution theory is then applied to approximations of American option prices under rough volatility in Section 3. Then in Section 4, we discuss the numerical approximations with a deep learning-based method: in the first subsection we addressed the approximations of neural networks to random functions involving infinite-dimensional input spaces in the probabilistic setting, then a deep learning-based method is introduced for non-Markovian BSDEs and associated BSPDEs in the second subsection, and in the third subsection we present some numerical examples for the rough Bergomi model. Finally, in the appendix, a convergence analysis is presented for the deep learning-based method.

\section{ Well-posedness of nonlinear BSPDEs and stochastic Feynman-Kac formula}

This section is devoted to a weak solution theory for the following nonlinear BSPDE:
\begin{equation}\label{BSPDE}
\left\{
\begin{split}
-du_t(x)&= \Big[\frac{V_t}{2}D^2u_t(x) + \rho \sqrt{V_t}D\psi_t(x)-\frac{V_t}{2}Du_t(x)\\
&\quad\quad+ F_t(e^x,u_t(x),\sqrt{(1-\rho^2)V_t}Du_t(x),\psi_t(x)+\rho \sqrt{V_t} Du_t(x))\Big]\,dt
		\\
		&\quad\quad
			-\psi_t(x)\,dW_s, \quad (t,x)\in [0,T)\times \bR;\\
u_T(x)&=G(e^x),\quad x\in\bR.
\end{split}\right.
\end{equation}
Noteworthily, BSPDE \eqref{BSPDE-value-funct} turns out to be a particular case when $F_t(x,y,z,\tilde z)\equiv -ry$ and $G(e^{x})=H(e^{x+rT})$.

We shall study the well-posedness of BSPDE \eqref{BSPDE} for given continuous nonnegative process $(V_t)_{t\geq 0}$ and address the representation relationship between BSPDE \eqref{BSPDE} and associated FBSDE.
Following are the assumptions on the coefficients $G$ and $F$.

\begin{ass} \label{ass}
\begin{enumerate} 
\item[(1)] The function $G:  (\Omega\times \bR,\,\sF_T^W\otimes \cB(\bR)) \rightarrow (\bR,\cB(\bR)$ satisfies
 \begin{align*}
 G(x) &\leq L(1+|x|),  \quad x\in\bR,
 \end{align*}
 for some constant $L>0$.
 \item[(2)] The function $F:(\Omega\times [0,T]\times \bR^4,\, \sP^W\otimes \cB(\bR^4))\rightarrow (\bR,\cB(\bR))$ satisfies that there exists a positive constants $L_0\in (0,\infty)$ such that  for all $x,y^1,y^2,z^1,z^2,\tilde z^1,\tilde z^2\in\bR$, and $t\in[0,T]$,
 \begin{align*}
 &|F_t(x,y^1,z^1,\tilde z^1)-F_t(x,y^1, z^2,\tilde z^2)|\leq L_0\left(
 |y^1-y^2|+|z^1-z^2|+|\tilde z^1-\tilde z^2| \right), \text{ a.s.},
 \\
 &|F_t(x,0,0,0)| \leq L_0(1+|x|), \text{ a.s.,}\\
 &|F_t(x,y^1,z^1,\tilde z^1)-F_t(x,y^1,0,0)|\leq L_0   , \text{ a.s..}
 \end{align*}
\end{enumerate}
\end{ass}

For the well-posedness of BSPDE \eqref{BSPDE} under Assumption \ref{ass}, the difficulty lies in the combination of the non-uniform-boundedness of $(V_t)_{t\in[0,T]}$ and the inintegrability of $G(e^x)$ and $F_t(e^x,y,z,\tilde z)$ w.r.t. $x$ on the whole space $\bR$.  Indeed, from the condition on $(V)_{t\geq 0}$ in Assumption \ref{ass}, we may conclude that $e^{X_s^{0,x}}$ is a positive local martingale and thus a supermartingale, satisfying $E[e^{X_t^{0,x}}]\leq e^x$ for instance; however, it is not appropriate to expect $E\left[\big| e^{X_t^{0,x}}\big|^p\right]<\infty$ for some $p>1$ without further restrictive assumptions (see \cite[Theorem 2]{gassiat2018martingale}).

The dependence of $F$ on $(Z,\tilde Z)$ is not necessary for the concerned
examples in this paper. We assume the Lipschitz continuity and boundedness in
$(Z,\tilde Z)$ for the reader's interests. In fact, for the well-posedness of the involved BSDEs and BSPDEs in the $L^1$ spaces, it is not appropriate to assume the linear growth in $(Z,\tilde Z)$ as indicated in the theory of $L^1$ solutions for BSDEs (see \cite[Section 6]{Hu_2002}); it might be workable for certain fractional growths in $(Z,\tilde Z)$, while we would not seek such a generality to avoid cumbersome arguments in this work.

Corresponding to BSPDE \eqref{BSPDE}, there follows the BSDE:
\begin{equation}\label{BSDE}
\left\{
\begin{split}
-dY_s^{t,x}&= F_s(e^{X_s^{t,x}},Y_s^{t,x},Z_s^{t,x},\tilde Z_s^{t,x})\,ds -
\tilde Z_s^{t,x} \,dW_s-  Z_s^{t,x} \,dB_s  , \quad 0\leq t\leq s\leq T;\\
Y_T^{t,x}&=G(X_T^{t,x}),
\end{split}\right.
\end{equation}
where the triple $(Y_s^{t,x},Z_s^{t,x},\tilde{Z}_s^{t,x})$ is defined as the
solution to BSDE \eqref{BSDE} in the sense of \cite[Definition
2.1]{Hu_2002}. Under Assumptions \ref{ass:stoch-var} and \ref{ass}, BSDE \eqref{BSDE} has a unique
solution $(Y_s^{t,x},Z_s^{t,x},\tilde{Z}_s^{t,x})$ for each
$(t,x)\in[0,T)\times \bR$ (see \cite[Theorem 6.3]{Hu_2002}).

\subsection{Definition of the weak solution for BSPDE \eqref{BSPDE}}
Denote by $C_c^{\infty}$ the space of infinitely differentiable functions with compact supports in $\bR$ and let $\mathscr{D}$ be the space of real-valued Schwartz distributions on
$C^{\infty}_{c}$.  The Lebesgue measure in $\bR$ will be denoted by $dx$. $L^2(\bR)$ ($L^2$ for short) is the usual Lebesgue integrable space with scalar product and norm defined
$$
\langle \phi,\,\psi\rangle=\int_{\bR}\phi(x)\psi(x)dx,\quad \|\phi\|=\langle\phi,\,\phi\rangle^{1/2},\,\,\forall
\phi,\psi\in L^2.
$$
For convenience, we shall also use $\langle \cdot,\,\cdot\rangle$ to denote the duality between the Schwartz distribution space $\mathscr{D}$ and  $C_c^{\infty}$.

By $\mathfrak{D}_{\sF}$ (respectively, $\mathfrak{D}_{\sF^W}$) we denote the set of all $\mathscr{D}$-valued functions defined on
$\Omega\times [0,T]$ such that, for any $u\in \mathfrak{D}_{\sF}$ (respectively, $u\in\mathfrak{D}_{\sF^W}$) and $\phi\in C_c^{\infty}$, the
function $\langle u,\,\phi\rangle$ is $\sP$ (respectively, ${\sP}^W$)-measurable. When there is no confusion about the involved filtration, we shall just write $\mathfrak{D}$.

For $p=1,2$ we denote by $\mathfrak{D}^p$ the totality of $u\in\mathfrak{D}$ such that for any $R_1\in(0,\infty)$ and $\phi\in C_c^{\infty}$, we have
$$
\int_0^{T} \sup_{|x|\leq R_1} |\langle u_t(\cdot),\phi(\cdot-x)\rangle|^p \,dt<\infty \quad \text{a.s..}
$$

\begin{lem}\label{lem-D}
Given $u\in \mathfrak D^p$ for $p=1,2$, it holds that:
\begin{enumerate}
\item[(i)] $Du\in \mathfrak D^p$;
\item[(ii)] For each continuous function $\varrho$  on $\bR$, 
 we have $\varrho u\in \mathfrak D^p$ if $u\in L^2(\Omega\times[0,T]\times \bR)$.
\item[(iii)] For any continuous processes $(x_t)_{t\in[0,T]}$ and $(y_t)_{t\in[0,T]}$ with $\max_{t\in[0,T]} |x_t| +|y_t| <\infty$ a.s., the random field $\tilde u_t(x):=y_tu_t(x+x_t)$ is also lying in $\mathfrak D^p$.
\end{enumerate}
\end{lem}
\begin{proof}
The assertion (i) may also be found in \cite[page 297]{Krylov_09}. In fact, for each $\phi\in C_c^{\infty}$, we have $D\phi \in C_c^{\infty}$, and the integration-by-parts formula indicates that
$$\langle D u_t(\cdot),\phi(\cdot-x)\rangle= - \langle u_t(\cdot),(D\phi)(\cdot-x)\rangle.
$$
Hence, $Du\in \mathfrak D^p$ if $u\in  \mathfrak D^p$.

For assertion (ii), notice that for each $\gamma\in (0,\infty)$,
\begin{align*}
\sup_{|x|\leq \gamma}| \langle \varrho(\cdot) u_t(\cdot),\phi(\cdot-x)\rangle |^p
&\leq      \|u_t\|^p  \| \phi  \|^p \max_{|x|\leq \gamma+R}|   \varrho(x) |^p,
\end{align*}
where we choose a sufficiently big $R>0$ so that the support of $\phi$ is contained in $[-R,R]$. Then it follows obviously that $\varrho u\in  \mathfrak D^p$.

Lastly, as $\max_{t\in[0,T]} |x_t| + |y_t| <\infty$ a.s. and for each $\gamma\in (0,\infty)$,
\begin{align*}
\sup_{|x|\leq \gamma}| \langle  y_t u_t(\cdot+x_t),\phi(\cdot-x)\rangle |^p
&= \sup_{|x|\leq \gamma}|   \langle u_t(\cdot),  y_t\phi(\cdot-x_t-x)\rangle |^p
\\
&\leq \sup_{|x|\leq \gamma+ \max_{t\in[0,T]} |x_t|}|    \langle u_t(\cdot), \phi(\cdot-x)\rangle |^p  \max_{t\in[0,T]} |y_t|^p,
\end{align*}
there holds assertion (iii).
\end{proof}

For $u,f,g\in \mathfrak{D}$, we say that the equality
\begin{equation*}
du_t(x)=f_t(x)\,dt+g_t(x)\,dW_t, \quad t\in [0,T],
\end{equation*}
holds in the sense of distribution if $f   \in \mathfrak{D}^1$, $g   \in \mathfrak{D}^2$ and for each $\phi\in C_c^{\infty}$, it holds a.s.,
$$
\langle u_t(\cdot),\,\phi\rangle=
\langle u_0(\cdot),\,\phi\rangle+\int_0^t\langle f_s(\cdot),\,\phi\rangle\,ds+\int_0^t\langle g_s(\cdot),\,\phi\rangle \,dW_s, \quad
	\forall \, t\in[0,T].
$$

\begin{defn}\label{defn-solution}
A pair $(u,\psi)\in \mathfrak D^1_{\sF^W} \times \mathfrak D^2_{\sF^W}$ is said to be a weak solution of BSPDE \eqref{BSPDE}, if
\begin{enumerate}
\item [(i)] $u_T(x)=G(e^x)$ a.s.;
\item [(ii)] for almost all $(\omega,t)\in\Omega\times[0,T]$, the functions
  $u_t(x), \sqrt{(1-\rho^2)V_t}D\psi_t(x)$, and
  $\rho \sqrt{V_t}Du_t(x)+\psi_t(x)$ are locally integrable in $x\in\bR$;\footnote{Here, by the local integrability of a function $g$ in $x\in\bR$ we mean that for each bounded measurable set $D\subset \bR$, it holds that the truncated function $g \cdot 1_{D}$ lies in $L^1(\bR)$.}
\item[(iii)] the equality
\begin{align*}
-du_t(x)=& \Big[\frac{V_t}{2}D^2u_t(x) + \rho \sqrt{V_t}D\psi_t(x)-\frac{V_t}{2}Du_t(x)\\
&
	+F_t(e^x,u_t(x),\sqrt{(1-\rho^2)V_t}Du_t(x),\psi_t(x)+\rho \sqrt{V_t} Du_t(x))
\Big]\,dt
		-\psi_t(x)\,dW_s,
\end{align*}
holds in the sense of distribution.
\end{enumerate}
\end{defn}
By Assumption \ref{ass}, the linear growth of $(G,F)$ w.r.t. $e^x$ produces the local integrability in $x\in\bR$. Therefore,
in Definition \ref{defn-solution} the local integrability is set for the weak solution, which does not just give a point-wise meaning of the compositions involved in function $F$ but also make the weak solution be potentially workable under Assumption \ref{ass} particularly encompassing the concerned examples in this paper. Obviously, it differs from the $L^p$ ($p\in(1,\infty]$)-integrability requirements for the weak or viscosity solutions in the existing BSPDE literature (see \cite{DuQiuTang10,Hu_Ma_Yong02,qiu2017viscosity,Zhou_92} for instance).

\subsection{Well-posedness of BSPDE \eqref{BSPDE} and the stochastic Feynman-Kac formula}
First comes a result about the measurability of $Y_t^{t,x}$ which basically states that the randomness from Wiener process $B$ is averaged out as the randomness of all the coefficients is only (explicitly) subject to the sub-filtration $\{\sF^W_t\}_{t\geq 0}$.
\begin{thm}\label{thm-main-measurability}
Under assumptions \ref{ass:stoch-var} and \ref{ass}, for each $(t,x)\in[0,T]\times \bR$,  let $(Y_s^{t,x},Z_s^{t,x},\tilde{Z}_s^{t,x})$ be the solution to BSDE \eqref{BSDE}. 
Then the value function:
$$
\Phi_t(x):= Y_t^{t,x} \quad \text{is just $\sF_t^W$-measurable.}
$$

\end{thm}

\begin{proof}
We shall adopt some techniques by Buckdahn and Li in \cite{buckdahnLi-2008-SDG-HJBI}. For the underlying probability space, w.l.o.g., we may take $\Omega=C([0,T];\bR^{2})=\Omega^W\times \Omega^B$, with ${\Omega^W}=C([0,T];\bR)$, ${\Omega^B}=C([0,T];\bR)$, and for each $\omega\in\Omega$, one has $\omega=({\omega^W},{\omega^B})$ with ${\omega^W}\in {\Omega^W}$ and ${\omega^B}\in {\Omega^B}$. And the two independent Wiener processes $W$ and $B$ may be defined on $\Omega^W$ and $\Omega^B$, respectively.

Set
\[
\bH=\left\{h;h(0)=0,\frac{dh}{dt}\in L^2(0,T;\bR)\right\},
\]
which is the Cameron-Martin space associated with the Wiener process ${B}$. For any $h\in\bH$, we define the translation operator $\tau_h:{\Omega}\rightarrow {\Omega}$, $\tau_h(({\omega^W},{\omega^B}))=({\omega^W},{\omega^B}+h)$ for ${\omega}=({\omega^W},{\omega^B})\in{\Omega}$. It is obvious that $\tau_h$ is a bijection and that it defines the probability transformation: $\left(\bP\circ \tau_h^{-1}\right)(d\omega)=\exp\{\frac{1}{2}\int_0^T|\frac{dh}{dt}|^2\,dt-\int_0^T\frac{dh}{dt}\,d{B}_t\}\bP(d\omega)$.

Fix some $(t,x)\in[0,T]\times\bR^d$ and set $\bH_t=\{h\in \bH\, \big|\, h(\cdot)=h(\cdot \wedge t) \}$. Recall
$$
X_T^{t,x}=x-\int_t^T \frac{V_s}{2}ds+\int_t^T \rho \sqrt{V_s}\, dW_s +\int_t^T \sqrt{(1-\rho^2)V_s} \, dB_s.
$$
By Girsanov theorem, it follows that    $X^{t,x}_T(\tau_h)=X^{t,x}_T$ for all $h\in \bH_t$, and thus, we have
 $\Phi_t(x)(\tau_h)=\Phi_t(x)$ $\bP$-a.s. for any $h\in\bH_t$. In particular, for any continuous and bounded function $\mathcal G$,
\begin{align*}
&E\left[  \mathcal G(\Phi_t(x))\exp\Big\{\int_0^T|\frac{dh}{ds}|^2\,ds-\frac{1}{2}\int_0^T\frac{dh}{ds}\,dB_s\Big\}  \right]\\
&=
E\left[  \mathcal G(\Phi_t(x))(\tau_{h})\exp\Big\{\int_0^T|\frac{dh}{ds}|^2\,ds-\frac{1}{2}\int_0^T\frac{dh}{ds}\,dB_s\Big\}  \right]
\\
&=E\left[  \mathcal G(\Phi_t(x))\right]
\\
&=E\left[  \mathcal G(\Phi_t(x))\right]
E\left[\exp\Big\{\int_0^T|\frac{dh}{ds}|^2\,ds-\frac{1}{2}\int_0^T\frac{dh}{ds}\,dB_s\Big\}  \right],
\end{align*}
which together with the arbitrariness of $(\mathcal G,h)$ implies that $\Phi_t(x)$ is just $\sF^W_t$-measurable.
\end{proof}

Following is the It\^o-Wentzell-Krylov formula.
\begin{lem}[Theorem 1 of \cite{Krylov_09}]\label{Ito-Wentzell}
Let $x_t$ be an $\bR$-valued predictable process of the following form
$$
x_t=\int_0^tb_s\,ds+\int_0^t\beta_s\,dW_s+\int_0^t\sigma_s\,dB_s,
$$
where $b$, $\sigma$ and $\beta$ are predictable processes such that for all $\omega\in\Omega$ and $s\in[0,T]$,
it holds that
$$
|\beta_s|+ |\sigma_s| <\infty \quad \hbox { \rm and } \quad \int_0^T \left( |b_t|+ |\beta_s|^2+ |\sigma_s|^2  \right) \,dt<\infty.
$$
  Assume that the equality
\begin{equation*}
du_t(x)=f_t(x)\,dt+g_t(x)\,dW_t, \quad t\in [0,T],
\end{equation*}
 holds in the sense of distribution and define
  $
  v_t(x):=u_t(x+x_t).
  $
  Then we have
  \begin{equation*}
    \begin{split}
      dv_t(x)
      =&\bigg(
      f_t(x+x_t)+ \frac{1}{2}(|\beta_t|^2+|\sigma_t|^2)D^2v_t(x)+\beta_t Dg_t(x+x_t)
      +b_tDv_t(x)\bigg)\,dt\\
      &+\left( g_t(x+x_t)+\beta_t Dv_t(x)\right)\,dW_t
      +\sigma_tDv_t(x)\,dB_t
      ,\quad t\in[0,T]
    \end{split}
  \end{equation*}
  holds in the sense of distribution.
\end{lem}
We note that in the It\^o-Wentzell formula by Krylov \cite[Theorem 1]{Krylov_09}, the Wiener process $(W_t)_{t\geq 0}$ may be general separable Hilbert space-valued and the process $(x_t)_{t\geq 0}$ may be multi-dimensional. An application of the above It\^o-Wentzell-Krylov formula gives the following stochastic Feynman-Kac formula that is the probabilistic representation of the weak solution to BSPDE \eqref{BSPDE} via the solution of associated BSDE \eqref{BSDE} coupled with the forward SDE \eqref{forward-X}.

\begin{thm}\label{thm-uniqueness}
Let Assumptions \ref{ass:stoch-var} and \ref{ass} hold. Let $(u,\psi)$ be a weak solution of BSPDE \eqref{BSPDE} such that there is $C_u\in (0,\infty)$ satisfying for each $t\in[0,T]$
\begin{align}
|u_t(x)| \leq C_u \left(1+e^x\right), \quad \text{for almost all }(\omega,x)\in\Omega\times \bR. \label{growth-u}
\end{align}
Then $(u,\psi)$ admits a version (denoted by itself) satisfying a.s.
$$
u_{\tau}(X_{\tau}^{t,x})=Y_{\tau}^{t,x},\quad \sqrt{(1-\rho^2)V_{\tau}} Du_{\tau}(X_{\tau}^{t,x}) =Z_{\tau}^{t,x}, \quad
\psi_{\tau}(X_{\tau}^{t,x})+\rho\sqrt{V_{\tau}}Du_{\tau}(X_{\tau}^{t,x})
= \tilde Z_{\tau}^{t,x},
$$
for $0\leq t\leq \tau\leq T$ and   $x\in\bR$, where $(Y_{\tau}^{t,x},Z_{\tau}^{t,x},\tilde Z_{\tau}^{t,x})$ is the unique solution to BSDE \eqref{BSDE}.
\end{thm}

\begin{proof}
For each $t\in [0,T)$, recall
$$
X_s^{t,x}=x-\int_t^s \frac{V_r}{2}dr+\int_t^s \rho \sqrt{V_r}dW_r +\int_t^s \sqrt{(1-\rho^2)V_r} dB_r,\quad t\leq s \leq T.
$$
Applying Lemma \ref{Ito-Wentzell} to $u$ over the interval $[t,T]$ yields that
\begin{align*}
du_s(X_s^{t,x})
&=
\left(\psi_s(X_s^{t,x})+\rho\sqrt{V_s}Du_s(X_s^{t,x})\right) dW_s
+\sqrt{(1-\rho^2)V_s}Du_s(X_s^{t,x}) dB_s
\\
-F_s&\left(e^{X_s^{t,x}},u_s(X_s^{t,x}),\sqrt{(1-\rho^2)V_s}Du_s(X_s^{t,x}), \psi_s(X_s^{t,x})+\rho\sqrt{V_s}Du_s(X_s^{t,x})\right)\,ds
, \quad s\in[t,T],
\end{align*}
holds in the sense of distribution with $u_T(X_T^{t,x}) = G(e^{X_T^{t,x}})$.


Notice that for all $\tau\in [t,T]$, we have $e^{X_{\tau}^{t,x}}\in L^1(\Omega,\bP)$ and $E\left[e^{X_{\tau}^{t,x}}\right] \leq e^x$. This together with Assumption \ref{ass} and relation \eqref{growth-u}, implies that $ u_{\tau}(X_{\tau}^{t,x}) \in L^1(\Omega,\bP)$ for all $\tau\in [t,T]$. Further, the uniqueness of $L^1$-solution for BSDEs (see \cite[Section 6]{Hu_2002}) yields a version of $(u,\psi)$ (denoted by itself) satisfying that a.s.
 $$
u_{\tau}(X_{\tau}^{t,x})=Y_{\tau}^{t,x},\quad \sqrt{(1-\rho^2)V_{\tau}} Du_{\tau}(X_{\tau}^{t,x}) =Z_{\tau}^{t,x}, \quad
\psi_{\tau}(X_{\tau}^{t,x})+\rho\sqrt{V_{\tau}}Du_{\tau}(X_{\tau}^{t,x})
= \tilde Z_{\tau}^{t,x},
$$
for $0\leq t\leq \tau\leq T$ and   $x\in\bR$, where $(Y_{\tau}^{t,x},Z_{\tau}^{t,x},\tilde Z_{\tau}^{t,x})$ is the unique solution to BSDE \eqref{BSDE}.
\end{proof}
From the proof, we may see that the growth condition \eqref{growth-u} confirms that the distribution-valued process $u$ is locally integrable and a.e. defined on $\Omega\times [0,T]\times \bR$ which means more than distributions. More importantly, it implies the integrability of $u(\tau,X_{\tau}^{t,x})$ which is needed for the uniqueness of solution to BSDEs. The growth condition \eqref{growth-u} may be relaxed; however, power growth condition like $|u_t(x)| \leq C \left(1+|e^{x}|^p\right)$ for some $p>1$ may fail to imply the integrability of $u(\tau,X_{\tau}^{t,x})$ (see \cite[Theorem 2]{gassiat2018martingale}). On the other hand, the stochastic Feynman-Kac formula in Theorem \ref{thm-uniqueness} actually implies the uniqueness of weak solution for BSPDE \eqref{BSPDE} which together with the existence is summarized in what follows.

\begin{thm}\label{thm-existence}
Under Assumptions \ref{ass:stoch-var} and \ref{ass}, suppose further that there is an infinitely differentiable function $\zeta$ such that $\zeta(x)>0$ for all $x\in\bR$ and
\begin{align}
G(e^{\cdot+X^{0,0}_T} ) \zeta(\cdot) \in L^2(\Omega,\sF_T;L^2(\bR)),
\quad
\zeta(\cdot) F_{\cdot}(e^{\cdot+X^{0,0}_{\cdot}},0,0,0)  \in L^2(\Omega\times [0,T];L^2(\bR)).
\label{ass-G}
\end{align}
Then  BSPDE \eqref{BSPDE} admits a unique weak solution $(u,\psi)$ such that there is $C_u\in (0,\infty)$ satisfying for each $t\in[0,T]$
\begin{align}
|u(t,x)| \leq C_u (1+e^x), \quad \text{for almost all }(\omega,x)\in\Omega\times \bR. \label{growth-u-ex}
\end{align}
\end{thm}

\begin{proof}
\textbf{Step 1} (Existence).
Put $\theta(x)=\frac{\zeta(x)}{(1+\zeta(x))(1+x^2)}$ for $x\in\bR$. The theory of Banach space-valued BSDEs in \cite[Section 3]{DuQiuTang10} may be extended to nonlinear cases under Lipschitz assumptions with the standard application of Picard iteration. In particular, for the case of Hilbert spaces, applying \cite[Theorem 3.1]{HuPeng} to the following  Hilbert space-valued BSDE (with a trivial operator $A=0$ therein):
\begin{align}
\tilde u_t(x)
&=G(e^{x+X_T^{0,0}} ) \theta(x)
+\int_t^T \theta(x)F_s(e^{x+X_s^{0,0}},(\theta(x))^{-1}\tilde u_s(x),(\theta(x))^{-1}\tilde \psi^B_s(x), (\theta(x))^{-1}\tilde \psi^W_s(x))\,ds
\nonumber\\
&\quad
-\int_t^T \tilde \psi^B_s(x) \, dB_s-\int_t^T \tilde \psi^W_s(x)\, dW_s,\quad t\in [0,T].
\label{bsde-tilde}
\end{align}
gives the solution of the triple of $L^2(\bR)$-valued $(\sF_t)$-adapted random fields
\begin{align}
(\tilde u, \tilde \psi^B,\tilde\psi^W)\in L^2(\Omega;C([0,T];L^2(\bR))) \times L^2(\Omega\times [0,T]\times\bR)\times L^2(\Omega\times [0,T]\times\bR). \label{BSDE-sol-Banach}
\end{align}
Obviously, we have $(\tilde u, \tilde \psi^B,\tilde \psi^W)\in \mathfrak D^1_{\sF} \times \mathfrak D^2_{\sF} \times \mathfrak D^2_{\sF}$, and thus by assertion (ii) of Lemma \ref{lem-D}, it holds that
$$
(\hat u, \hat \psi^B,\hat \psi^W):=\frac{(\tilde u, \tilde \psi^B,\tilde \psi^W)}{\theta}  \in \mathfrak D^1_{\sF} \times \mathfrak D^2_{\sF} \times \mathfrak D^2_{\sF},
$$
satisfying BSDE:
\begin{align*}
\hat u_t(x)&=G(e^{x+X_T^{0,0}} )
+\int_t^T F_s(e^{x+X_s^{0,0}}, \hat u_s(x), \hat  \psi^B_s(x), \hat \psi^W_s(x))\,ds
-\int_t^T \hat \psi^B_s(x) \, dB_s
\\
&\quad
-\int_t^T \hat \psi^W_s(x)\, dW_s,\quad t\in [0,T].
\end{align*}
Also, it is straightforward to have that
\begin{align}
\hat u_t(x)= Y_t^{t,x+X_t^{0,0}}\quad \text{a.s., for all }(t,x)\in[0,T]\times \bR, \label{re-hatu}
\end{align}
with the triple $(Y_s^{t,x},Z^{t,x}_s,\tilde{Z}^{t,x}_s)_{s\in[t,T]}$ satisfying BSDE \eqref{BSDE}.

By Lemma \ref{lem-D}, we may apply the It\^o-Wentzell-Krylov formula in Lemma \ref{Ito-Wentzell} which yields that the equality
\begin{align}
&-d\hat u_t(x-X_t^{0,0})\\
&=
\bigg\{
-\frac{V_t}{2}D^2\hat u_t(x-X_t^{0,0})+\sqrt{(1-\rho^2)V_t}D\hat\psi^B_t(x-X_t^{0,0})
	+\rho \sqrt{V_t}D\hat\psi^W_t(x-X_t^{0,0})
\nonumber\\
&\quad\quad
	-\frac{V_t}{2} D\hat u_t(x-X_t^{0,0})
	+F_t(e^{x}, \hat u_t(x-X_t^{0,0}),   \hat\psi^B_t(x-X_t^{0,0}) , \hat\psi^W_t(x-X_t^{0,0}))
\bigg\}dt
\nonumber\\
&\quad\quad
-\left(\hat\psi^W_t(x-X_t^{0,0})-\rho \sqrt{V_t} D\hat u_t(x-X_t^{0,0}) \right)dW_t
\nonumber\\
&\quad\quad
-\left(\hat\psi^B_t(x-X_t^{0,0})-\sqrt{(1-\rho^2)V_t} D\hat u_t(x-X_t^{0,0}) \right)dB_t,\quad t\in [0,T],
\label{IWK-hatu}
\end{align}
holds in the sense of distribution. Notice that the equality \eqref{re-hatu} indicates that  for each $s\in[0,T]$
\begin{align}
\hat u_s(x-X_s^{0,0})=Y_s^{s,x}, \label{re-u-x}
\end{align}
which is just $\sF_s^W$-measurable by Theorem \ref{thm-main-measurability}. Thus, the stochastic integration w.r.t. $B$ should be vanishing, i.e., we have
$$\hat\psi^B_t(x)-\sqrt{(1-\rho^2)V_t} D\hat u_t(x)=0, \text{ a.s. for all }(t,x)\in [0,T]\times\bR.
$$
 Put
\begin{align*}
u_t(x)=\hat u_t(x-X_t^{0,0}) \quad \text{and}\quad \psi_t(x)=\hat\psi^W_t(x-X_t^{0,0})-\rho \sqrt{V_t} D\hat u_t(x-X_t^{0,0}), \quad (t,x)\in[0,T]\times \bR.
\end{align*}
The $\sF_t^W$-adaptedness of $u_t(x)$, and
the assertions (i) and (iii) of Lemma \ref{lem-D} imply $(u,\psi)\in \mathfrak D^1_{\sF^W}\times \mathfrak D^2_{\sF^W}$, and  the equality \eqref{re-hatu} writes equivalently
\begin{align*}
-du_t(x)
&=\bigg\{
\frac{V_t}{2}D^2u_t(x)
+\rho \sqrt{V_t}D \psi_t(x)
-\frac{V_t}{2} Du_t(x)
\\
&+F_t(e^x,u_t(x),\sqrt{(1-\rho^2)V_t}Du_t(x),\psi_t(x)+\rho\sqrt{V_t}Du_t(x))
\bigg\}dt
-\psi_t(x)\,dW_t, \quad t\in[0,T],
\end{align*}
which holds in the sense of distribution with the terminal condition $u_T(x)=G(e^x)$.
 The local integrability of $(u,\sqrt{(1-\rho^2)V_t}Du,\psi+\rho\sqrt{V_t}Du)$ required in Definition \ref{defn-solution} (ii) may be obtained by combining the relation \eqref{BSDE-sol-Banach}, the path-continuity of $(X^{0,0}_{s})_{s\geq 0}$, and the positivity of $\theta$.
 Therefore, the pair $(u,\psi)$ is a weak solution of BSPDE \eqref{BSPDE}.

 \textbf{Step 2} (Growth condition \eqref{growth-u-ex}).
 Consider the following Hilbert space-valued BSDE:
 \begin{equation}\label{bsde-comparison}
 \left\{
 \begin{split}
 \tilde Y_t(x)
 &=\left|G(e^{x+X_T^{0,0}} )\right| \theta(x)
 -\int_t^T \tilde Z^B_s(x) \, dB_s-\int_t^T \tilde Z^W_s(x)\, dW_s
 \\
&
+\int_t^T  \left( \theta(x) \big|F_s(e^{x+X_s^{0,0}},0,0,0)\big|+L_0\theta(x)+ L_0 |\tilde Y_s(x)| \right) \,ds,
 \end{split}\right.
 \end{equation}
 where the positive constant $L_0$ is from Assumption \ref{ass} (iii). The standard BSDE theory (see \cite{ParPeng_90}) yields the unique existence of the $L^2$-solution to BSDE \eqref{bsde-comparison}. In fact, for each $(t,x)\in [0,T)\times \bR$ we have
 \begin{align}
 &\tilde Y_t(x)= E\left[ \left|G(e^{x+X_T^{0,0}} )\right| \theta(x) \gamma_T^t + \int_t^T \theta(x) (L_0+\big|F_s(e^{x+X_s^{0,0}},0,0,0)\big|) \cdot \gamma_s^t \,ds \Big| \sF_t \right],
 \label{Y-tilde-repr}
 \\
 &\text{with }\gamma_s^t= \exp\left\{ L_0 (s-t)   \right\},\quad s\in[t,T]. \nonumber
 \end{align}
 Putting the BSDEs \eqref{bsde-tilde} and \eqref{bsde-comparison} together, we may use the comparison theorem (see \cite[Theorem 2.2]{Karoui_Peng_Quenez}) to achieve the relation
 \begin{align*}
 \tilde u_t(x)\leq \tilde Y_t(x), \quad \text{a.s., }\forall (t,x)\in[0,T]\times \bR,
 \end{align*}
 which together with \eqref{Y-tilde-repr} implies that
 \begin{align*}
 u_t(x)
 &\leq (\theta(x-X_t^{0,0}))^{-1} \tilde Y_t(x-X_t^{0,0})
 \\
 &=E\left[ \left|G(e^{X_T^{t,x}} )\right|  \gamma_T^t + \int_t^T \left(\big|F_s(e^{X_s^{t,x}},0,0,0)\big| +L_0\right) \cdot \gamma_s^t \,ds \Big| \sF_t \right]
 \\
 &\leq
 E\left[ \left(L e^{X_T^{t,x}} +L \right) \gamma_T^t + \int_t^T \left( L_0e^{X_s^{t,x}} +2L_0 \right) \cdot \gamma_s^t \,ds \Big| \sF_t \right]
 \\
 &\leq
 E\left[ L e^{x+L_0(T-t)}+L e^{L_0(T-t)} + \int_t^T \left( L_0e^{x+L_0(s-t)} +2L_0e^{L_0(s-t)}\right)   \,ds \Big| \sF_t \right]
 \\
 &\leq C(L,T,L_0) (1+e^x),\text{ a.s., }\forall (t,x)\in [0,T]\times \bR,
 \end{align*}
 where we have used the relation $E\left[ e^{X_s^{t,x}}\big|\sF_t\right] \leq e^x$ a.s., for $0\leq t\leq s\leq T$. This gives the growth estimate \eqref{growth-u-ex}.

\textbf{Step 3} (Uniqueness).
The uniqueness follows from Theorem \ref{thm-uniqueness} and the proof is complete.
\end{proof}

\begin{rmk}
In view of the above proof, the assumption \eqref{ass-G} on $G$ and $F$ is to ensure $(\tilde \psi^B, \tilde \psi^W)\in \mathfrak D^2_{\sF} \times \mathfrak D^2_{\sF}$ and further $\psi\in \mathfrak D^2_{\sF^W}$. It is for simplicity and may be relaxed; for instance, the $L^2$-requirements in \eqref{ass-G} may be replaced correspondingly by $L^p$-integrability but with $1<p<\infty$ and the associated well-posedness result with $L^p$-integrability in \eqref{BSDE-sol-Banach} may be obtained by standardly extending the theory of Banach space-valued BSDEs in  \cite[Section 3]{DuQiuTang10} as stated at the beginning of the proof. A typical example satisfying \eqref{ass-G} is the European put option where $F_t(x,y,z,\tilde z)=-ry$, $G(e^x)=(K-e^{x+rT})^+$ for some $K\in(0,\infty)$ and one may take $\zeta(x)=\frac{1}{1+x^2}$ for instance.  However, it is by no means obvious to see if it is satisfied for the call options, while for pricing calls, we may use the put-call parity  if applicable.
\end{rmk}



\section{An application: approximating American option prices}\label{section:AO}

Assuming the same setting as the European options, we consider instead the American type, that is to compute
\begin{align*}
\overline u_t(x):= \sup_{\tau\in \mathcal T_t}E\left[   e^{-(\tau-t) r}g_{\tau}(e^{X_{\tau}^{t,x} })\big| \sF_t  \right], \quad (t,x) \in [0,T]\times \bR,
\end{align*}
where $r\geq 0$ is the interest rate and $\mathcal T_t$ denotes all the stopping times $\tau$ satisfying $t\leq \tau \leq T$. For simplicity, we assume:
\begin{ass}\label{ass-AO}
The function $g:(\Omega\times [0,T]\times \bR,\, \sP^W\otimes \cB(\bR))\rightarrow (\bR,\cB(\bR))$ satisfies that there exists a positive constant $L_1>0$ such that for each $(t,x)\in [0,T]\times \bR$,
\begin{enumerate}
\item [(i)] $g_s(e^{X_s^{t,x}})$ is almost surely continuous in $s\in[t,T]$;
\item [(ii)] $g_s(e^{x})\leq L_1(1+e^x)$, a.s.;
\item [(iii)]
 \begin{align*}
&\left|g_s\left(e^{X_s^{t,x}}\right)\right|
\leq \Gamma_s^t \tilde \theta(x),\quad\text{a.s., } \forall\, s\in[t,T],\text{ with }E\left[ \sup_{s\in[t,T]} \left|\Gamma_s^t \right|^2\right]<\infty,
 \end{align*}
where the positive function $\tilde \theta: \bR\rightarrow (0,\infty)$  is infinitely differentiable.

\end{enumerate}
\end{ass}
A typical example satisfying Assumption \ref{ass-AO} is the American put option with $g_t(e^{x}) =(K-e^{x+rt})^+$ for some $K>0$, where one may take $L_1=K$, $\Gamma_s^t \equiv K$,  and $\tilde \theta(x) \equiv 1$.
By the theory of reflected BSDEs (see \cite[Section 3]{El_Karoui-reflec-1997}), the following reflected BSDE
\begin{equation}\label{RBSDE}
\left\{
\begin{split}
-d\overline Y_s^{t,x}&=-r\overline Y_s^{t,x}\,ds +dA^{t,x}_s-\overline Z_s^{t,x;B}\,dB_s -\overline Z_s^{t,x;W}\,dW_s,\quad s\in[t,T];\\
\overline Y_T^{t,x}&=g_T(e^{X_T^{t,x}});\quad \overline Y_s^{t,x}\geq g_s(e^{X_s^{t,x}}),\quad s\in [t,T];\\
A^{t,x}_{\cdot}& \text{ is increasing and continuous, }A_t^{t,x}=0,
\quad \int_t^T (\overline Y_s^{t,x}-g_s(e^{X_s^{t,x}}))\,d A_s^{t,x}=0,
\end{split}\right.
\end{equation}
admits a unique solution $(\overline Y^{t,x},A^{t,x},\overline Z^{t,x;B}, \overline Z^{t,x;W})$ for each $(t,x)\in [0,T]\times \bR$, and in particular, by \cite[Proposition 7.1]{El_Karoui-reflec-1997}, we have
\begin{align}
\overline Y_t^{t,x}= \overline u_t(x), \text{ a.s. for each }(t,x)\in [0,T]\times \bR. \label{RBSDE-AO}
\end{align}
We would stress that the above relation \eqref{RBSDE-AO} only indicates that $\overline u_t(x)$ is $\sF_t$-measuable for each $(t,x)\in [0,T]\times \bR$.

In fact, the penalization method provides an approximation of reflected BSDE \eqref{RBSDE} with a sequence of BSDEs without reflections (see \cite[Section 6]{El_Karoui-reflec-1997}), i.e., for each $N\in \bN^+$, the following BSDE
\begin{equation}\label{RBSDE-appr}
\left\{
\begin{split}
-d\overline Y_s^{t,x;N}&= \left[ -r\overline Y_s^{t,x;N}+N\left(g_s(e^{X_s^{t,x}})-\overline Y^{t,x;N}_s\right)^+ \right]\,ds
-\overline Z_s^{t,x;B,N}\,dB_s \\
&\quad \quad-\overline Z_s^{t,x;W,N}\,dW_s,\quad s\in[t,T];\\
\overline Y_T^{t,x;N}&=g_T(e^{X_T^{t,x}}),
\end{split}\right.
\end{equation}
admits a unique solution $(\overline Y^{t,x;N},\overline Z^{t,x;B,N}, \overline Z^{t,x;W,N})$ such that
$\overline Y^{t,x;N}_s$ converges increasingly to $\overline Y^{t,x}_s$ with
\begin{align}
\lim_{N\rightarrow \infty}
E\left[\sup_{s\in[t,T]} \left|
\overline Y^{t,x;N}_s
-\overline Y^{t,x}_s
\right|^2
+\int_t^T \left| \overline Z^{t,x;B,N}_s
-\overline Z^{t,x;B}_s \right|^2
+   \left| \overline Z^{t,x;W,N}_s
-\overline Z^{t,x;W}_s \right|^2
ds
\right] =0, \label{BSDE-appr-rbsde}\\
\lim_{N\rightarrow \infty}
E\left[\sup_{s\in[t,T]} \left|
 A^{t,x;N}_s
- A^{t,x}_s
\right|^2
 \right] =0, \label{BSDE-appr-rbsde-1}
\end{align}
for each $(t,x)\in[0,T]\times \bR$, where
$$A^{t,x;N}_r=\int_t^r N\left(g_s(e^{X_s^{t,x}})-\overline Y^{t,x;N}_s\right)^+ \,ds, \quad \text{for }0\leq t\leq r\leq T.$$
Notice that Theorem \ref{thm-main-measurability} says that $\overline Y^{t,x;N}_t$ is $\sF^W_t$-measurable for each $(t,x)\in [0,T]\times \bR$. Hence, the approximation \eqref{BSDE-appr-rbsde} implies that $\overline Y_t^{t,x}$ (and thus $\overline u_t(x)$) is also just  $\sF^W_t$-measurable for each $(t,x)\in [0,T]\times \bR$, which together with Theorems \ref{thm-uniqueness} and \ref{thm-existence} yields the following
\begin{cor}\label{cor-AO}
Let Assumptions \ref{ass:stoch-var} and \ref{ass-AO} hold. It holds that:
\begin{enumerate}
\item [(i)] The value function $\overline u_t(x)$ is just  $\sF^W_t$-measurable for each $(t,x)\in [0,T]\times \bR$.
\item [(ii)] For each $N\in\bN^+$, the following BSPDE
\begin{equation}\label{BSPDE_N}
\left\{
\begin{split}
-du^N_t(x)&= \Big[\frac{V_t}{2}D^2u^N_t(x) + \rho \sqrt{V_t}D\psi^N_t(x)-\frac{V_t}{2}Du^N_t(x)-ru^N_t(x)
\\
&\quad\quad
+N\left(g_t(e^{x})-u^N_t(x)\right)^+\Big]\,dt
			-\psi^N_t(x)\,dW_s;\\
u^N_T(x)&=g_T(e^{x}),
\end{split}\right.
\end{equation}
admits a unique weak solution $(u^N,\psi^N)$ such that there exists $C_N\in (0,\infty)$ satisfying for each $t\in[0,T]$
\begin{align*}
|u^N_t(x)| \leq C_N (1+e^x), \quad \text{for almost all }(\omega,x)\in\Omega\times \bR.
\end{align*}
\item [(iii)] For each $N\in\bN^+$, the above weak solution $(u^N,\psi^N)$
satisfies a.s. $u^N_{\tau}(X_{\tau}^{t,x})=\overline Y_{\tau}^{t,x;N}$,
 $$
 \sqrt{(1-\rho^2)V_{\tau}} Du^N_{\tau}(X_{\tau}^{t,x}) =\overline Z_{\tau}^{t,x;B,N}, \quad \text{and}\quad
\psi^N_{\tau}(X_{\tau}^{t,x})+\rho\sqrt{V_{\tau}}Du^N_{\tau}(X_{\tau}^{t,x})
= \overline Z_{\tau}^{t,x;W,N},
$$
for $0\leq t\leq \tau\leq T$ and   $x\in\bR$, where $(\overline Y^{t,x;N},\overline Z^{t,x;B,N},\overline Z^{t,x;W,N})$ is the unique solution to BSDE \eqref{RBSDE-appr}.
\item [(iv)] For each $(t,x)\in [0,T]\times \bR$, $u^N_t(x)$ converges increasingly to $\overline u_t(x)$ in $L^2(\Omega,\sF_t;\bR)$.
\item[(v)] There is a triple $(\overline u,\, \overline\psi^B,\,\overline\psi^W)$ defined on $(\Omega\times [0,T]\times \bR,\sP^W\otimes \cB(\bR))$ such that
\begin{align*}
\overline u_{\tau}(X_{\tau}^{t,x})=\overline Y_{\tau}^{t,x},\quad
\overline \psi^B_{\tau}(X_{\tau}^{t,x}) =\overline Z_{\tau}^{t,x;B}, \quad \text{and}\quad
\overline \psi^W_{\tau}(X_{\tau}^{t,x})
= \overline Z_{\tau}^{t,x;W},\quad \text{a.s.},
\end{align*}
for $0\leq t\leq\tau \leq T$.
\end{enumerate}
\end{cor}

\begin{rmk}\label{relfected-BSPDE}
The assertion (v) is concluded from the approximating relations \eqref{BSDE-appr-rbsde} and \eqref{BSDE-appr-rbsde-1}. In fact, by the theory of reflected BSPDEs (see \cite{QiuWei-RBSPDE-2013} or \cite[Section 3.3]{Qiu2014weak}), one may expect the value function $\overline u_t(x)$ to be characterized via the following reflected BSPDE
 \begin{equation}\label{RBSPDE}
  \left\{\begin{array}{l}
  \begin{split}
  -d\overline u_t(x)
  =\,&\displaystyle \Bigl[ \frac{V_t}{2}D^2\overline u_t(x) + \rho \sqrt{V_t}D\overline \psi_t(x)-\frac{V_t}{2}D\overline u_t(x)-r\overline u_t(x)
                \Bigr]\, dt+\overline \mu(dt,x)\\ &\displaystyle
           -\overline \psi(t,x)\, dW_{t},\,(t,x)\in [0,T]\times \bR;\\
    \overline u_T(x)=\, &g_T(e^{x}), \quad x\in\bR;\\
    \overline u_t(x)\geq\,& g_t(e^{x}),\,\,d\mathbb{P}\otimes dt\otimes dx\text{-a.e.};\\
    \int_{[0,T]\times \bR} \big( \overline u_t(x)&-g_t(e^{x}) \big)\,\overline\mu(dt,dx)=0,\,\text{a.s.,} \quad \quad \textrm{(Skorohod condition)}
    \end{split}
  \end{array}\right.
\end{equation}
for which the solution is a triple $(\overline u, \overline \psi, \overline \mu)$ with $\overline\mu$ being a regular random radon measure. A solution theory may be developed by generalizing the regular stochastic potential and capacity theory in \cite{Qiu2014weak,QiuWei-RBSPDE-2013}; nevertheless, we would not seek such a generality in this paper, in order to put more efforts in the numerical approximations.
\end{rmk}



\section{Numerical approximations with a deep learning-based method}
Throughout this section, we assume that the functions $G$, $F$ and $g$ are deterministic, i.e., 
$$(\cA^*)\quad 
G: \bR\rightarrow \bR,\quad
F: [0,T]\times \bR^4 \rightarrow \bR,\quad g: [0,T]\times \bR\rightarrow \bR.$$
In fact, this assumption may be relaxed by allowing (explicit) dependence on the variance process $V$ and the Wiener process $W$, and together with Assumptions \ref{ass:stoch-var}, \ref{ass}, and \ref{ass-AO}, it ensures that all the coefficients may be simulated in the subsequent numerical computations, given the approximations of the unknown functions. In what follows, we first introduce and discuss the neural networks approximating random functions, a deep learning-based method is then introduced for non-Markovian BSDEs and associated BSPDEs
and finally, the numerical examples are presented for the rough Bergomi model.

\subsection{Neural networks approximating random functions}
First, we introduce a feedforward neural network with input dimension $d_0$ and output dimension $d_1$. Suppose that it has $M+1\in\mathbb{N}^+\diagdown\{1,2\}$ layers with each layer having $m_n$ neurons, $n=0,\cdots,M$.  For simplicity, we choose an identical number of neurons for all hidden layers, i.e., $m_n=m,n=1,\cdots,M-1$. Obviously, we have  $m_0=d_0$, and $m_M=d_1$. The neural network may be thought of as a function from $\bR^{d_0}$ to $\bR^{d_1}$ defined by composition of simple functions as
\begin{equation}\label{nnfunc}
x\in\bR^{d_0}\mapsto A_M\;o\;\varrho\;o\;A_{M-1}\;o\cdots o\;\varrho\;o\;A_1(x)\in\bR^{d_1}.
\end{equation}
Here, $A_1:\bR^{d_0}\mapsto\bR^m,A_M:\bR^{m}\mapsto\bR^{d_1}\;$ and $A_n:\bR^{m}\mapsto\bR^m,n=2,\cdots,M-1$ are affine transformations on a whole layer and defined by
\begin{equation*}
A_n(x)=\mathcal{W}_nx+\beta_n,
\end{equation*}
where the matrix $\mathcal{W}_n$ and the vector $\beta_n$ are called weight and bias respectively for the $n$th layer of the network. For the last layer we choose identity function as activation function, and the activation function $\varrho$ is applied component-wise on the outputs of $A_n$, for $n=1,\dots, M-1$.

The parameters of neural network may be denoted by $\theta=(\mathcal{W}_n,\beta_n)_{n=1}^M$. Given $d_0,d_1,M$ and $m$, the total number of parameters in a network is $M_m=\sum_{n=0}^{M-1}(m_n+1)m_{n+1}=(d_0+1)m+(m+1)m(M-1)+(m+1)d_1$ and thus $\theta\in\bR^{M_m}$. By $\Theta_m$, we denote the set of all possible parameters and if there are no constraints on parameters, we have $\Theta_m=\bR^{M_m}$. By $\Phi_m(\cdot;\theta)$ we denote the neural network function defined in (\ref{nnfunc}) and set of all such neural networks $\Phi_m(\cdot;\theta),\theta\in\Theta_m$ is denoted by $\mathcal{NN}_{d_0,d_1,M,m}^{\varrho}(\Theta_m)$.
\par Deep neural networks may approximate large classes of unknown functions. Following is a fundamental result by Hornik et al. \cite{hornik1989multilayer,hornik1990universal}:
\begin{lem}[Universal Approximation Theorem] \label{lem-univ-approx}
It holds that:
\begin{enumerate}
 \item[(i)] For each $M\in\bN^+\setminus \{1\}$, the set $\cup_{m\in\mathbb{N}}\mathcal{NN}_{d_0,d_1,M,m}^{\varrho}(\bR^{M_m})$ is dense in $L^2(\bR^{d_0},\nu(dx);\bR^{d_1})$ for any finite measure $\nu$ on $\bR^{d_0}$, whenever $\varrho$ is continuous and non-constant.
 \item[(ii)]     Assume that $\varrho$ is a non-constant $C^k$ function. Then the neural networks/functions in $\cup_{m\in\mathbb{N}}\mathcal{NN}_{d_0,d_1,2,m}^{\varrho}(\bR^{2_m})$ can approximate any function and its derivatives up to order $k$, arbitrarily well on any compact set of $\bR^{d_0}$.
 \end{enumerate}
\end{lem}
Notice that in the above lemma the approximated functions are defined on the finite dimensional spaces i.e., $\bR^{d_0}$. In fact, the approximations may be extended to some classes of functions defined on infinite dimensional spaces. In this paper, we need the following one:
\begin{prop}\label{prop-infty}
For each $T_0\in (0,T]$, $M\in\bN^+\setminus \{1\}$, and $d_0,d_1\in \bN^+$, the function set
\begin{align*}
\Big\{
\Phi_m(W_{t_1},\cdots,W_{t_k},x;\theta):\,\, \Phi_m(\cdot;\theta)\in \mathcal{NN}_{d_0+k,d_1,M,m}^{\varrho}(\bR^{M_m}), \, m,k\in\bN^+ ,\quad \\
0<t_1<t_2<\cdots<t_k\leq T_0
\Big\}
\end{align*}
is dense in
 $L^2\left(\Omega\times \bR^{d_0}, \sF^W_{T_0}\otimes \cB(\bR^{d_0}), \mathbb P(d\omega)\otimes dx;\bR^{d_1}\right)$, whenever $\varrho$ is continuous and non-constant.
\end{prop}
\begin{proof}
Take $f\in L^2\left(\Omega\times \bR^{d_0}, \sF^W_{T_0}\otimes \cB(\bR^{d_0}), \mathbb P(d\omega)\otimes dx\right)$  arbitrarily. Notice that
$$L^2\left(\Omega\times \bR^{d_0}, \sF^W_{T_0}\otimes \cB(\bR^{d_0}), \mathbb P(d\omega)\otimes dx;\bR^{d_1}\right)\equiv L^2\left(\Omega, \sF^W_{T_0}, \mathbb P; L^2(\bR^{d_0};\bR^{d_1}) \right).$$ The denseness of simple random variables (see \cite[Lemma 1.2, Page 16]{da2014stochastic} for instance) implies that the function $f$ may be approximated \textit{monotonically}  by simple random variables of the following form:
  $$
  \sum_{i=1}^{l}1_{A_i} (\omega) h_i(x),\quad \text{with }h_i\in L^2(\bR^{d_0};\bR^{d_1}),\quad A_i\in\sF_{T_0}, \quad l\in\bN^+, \quad i=1,\dots, l.
  $$

  Further,  applying \cite[Lemma 4.3.1., page 50]{oksendal2003stochastic}  yields that each $1_{A_i}$ may be approximated in $L^2(\Omega,\sF_{T_0})$ by functions in the following set
  $$
  \{g_i(W_{\tilde t^i_1},\dots,W_{\tilde t^i_{k_i}} ):\, \,k_i\in\bN^+,\,g_i\in  C^{\infty}_{c}(\bR^{k_i}),\,\, 0<\tilde t^i_1<\cdots<\tilde t^i_{k_i}\leq T_0 \}.
  $$

  To sum up, the function $f$ may be approximated in {\small $L^2\left(\Omega\times \bR^{d_0}, \sF^W_{T_0}\otimes \cB(\bR^{d_0}), \mathbb P(d\omega)\otimes dx;\bR^{d_1}\right)$} by the following random fields:
  \begin{align*}
  f^k( W_{\bar t_1},\cdots, W_{\bar t_k},x)
  = \sum_{i=1}^{l} g_i\left(W_{\tilde t^i_1},\dots,W_{\tilde t^i_{k_i}} \right)h_i(x),
    \end{align*}
  where $g_i \in C^{\infty}_{c}(\bR^{k_i})$, $h_i\in L^2(\bR^{d_0})$,  $0 <\bar t_1<\cdots<\bar t_{k}\leq T$, and
  $$\{\bar t_1,\ldots,\bar t_k\}=\cup_{i=1}^l\{\tilde t_1^i,\ldots,\tilde t_{k_i}^i     \}.$$
   Applying the approximation in (i) of Lemma \ref{lem-univ-approx} to the functions $f^k$ yields the approximation of $f$, and this completes the proof.
\end{proof}
\begin{rmk}\label{rmk-filtration}
In fact, the process $(W_t)_{t\geq 0}$ and the filtration $(\sF^{ W}_t)_{t\geq 0}$ may be replaced by an arbitrary continuous process $(\overline W_t)_{t\geq 0}$ and corresponding gernerated filtration $(\sF^{\overline W}_t)_{t\geq 0}$, where the process $(\overline W_t)_{t\geq 0}$ is  not necessarily a Brownian motion.
\end{rmk}

\subsection{Deep learning-based method for non-Markovian BSDEs and associated BSPDEs} \label{Sec:deep-learning}
Inspired by \cite{hure2019some,han2018solving},  we adopt a deep learning method  based on the following representation relationship by Theorems \ref{thm-uniqueness} and \ref{thm-existence}. Letting the quadruple $(X_s,Y_s,Z_s,\tilde{Z}_s)$ be the solution to the following FBSDE
\begin{equation}\label{FBSDE}
\left\{
\begin{split}
-dY_s&= F_s(e^{X_s},Y_s ,Z_s,\tilde Z_s)\,ds -
\tilde Z_s \,dW_s-  Z_s \,dB_s  , \quad 0 \leq s\leq T;\\
Y_T &=G(e^{X_T}),\\
dX_s&= \sqrt{V_s} \left( \rho \,dW_s+ \sqrt{1-\rho^2} \,dB_s  \right) -\frac{V_s}{2}\,ds, \quad 0 \leq s\leq T;\\
X_0&=x;\\
V_s&=\xi_s \, \cE(\eta\, \widehat W_s) \quad \text{with} \quad \widehat W_s= \int_0^s \mathcal K(s,r)\,dW_r, \quad s\in[0,T],
\end{split}\right.
\end{equation}
with $\mathcal K$ being a general Kernel function including the particular cases in Examples \ref{ex:rBergomi} and \ref{ex:rough-Heston},
one has
$$
u_{\tau}(X_{\tau})=Y_{\tau},\quad \sqrt{(1-\rho^2)V_{\tau}} Du_{\tau}(X_{\tau} ) =Z_{\tau}, \quad
\psi_{\tau}(X_{\tau} )+\rho\sqrt{V_{\tau}}Du_{\tau}(X_{\tau})
= \tilde Z_{\tau},
$$
for $0\leq \tau\leq T$ and   $x\in\bR$, where the pair $(u,\psi)$ is the unique weak solution to BSPDE \eqref{BSPDE} in Theorem \ref{thm-existence}. In particular, we may write forwardly, for $t\in[0,T]$,
\begin{align}
u_t(X_t)
&= u_0(X_0) -\int_0^t F_s\left( e^{X_s},u_s(X_s),\sqrt{(1-\rho^2)V_{s}} Du_{s}(X_{s} ), \psi_{s}(X_{s} )+\rho\sqrt{V_{s}}Du_{s}(X_{s})  \right)\,ds
	\nonumber\\
&\quad
	+\int_0^t  \left( \psi_{s}(X_{s} )+\rho\sqrt{V_{s}}Du_{s}(X_{s}) \right) \,dW_s
	+\int_0^t \sqrt{(1-\rho^2)V_{s}} Du_{s}(X_{s} )\,dB_s.
	\label{forward-u}
\end{align}

Given a partition of the time interval: $\pi=\{0=t_0<t_1<...<t_N=T\}$ with modulus $|\pi|=\max\limits_{i=0,1,...,N-1}\Delta t_i$, $\Delta t_i=t_{i+1}-t_i$, we first simulate (or approximate) the joint process $(B,W,V)$, and then the forward process $X$ may be approximated by $X^{\pi}$ obtained through an Euler scheme. Further,  the forward representation \eqref{forward-u} yields an approximation for $(u,\psi)$ under the Euler scheme
\begin{align*}
u_{t_{i+1}}(X_{t_{i+1}})
&
\approx H_{t_i}(X_{t_i},u_{t_i}(X_{t_i}),\sqrt{(1-\rho^2)V_{t_i}} Du_{t_i}(X_{t_i} ), \psi_{t_i}(X_{t_i} )+\rho\sqrt{V_{t_i}}Du_{t_i}(X_{t_i}),\Delta B_{t_i},\Delta W_{t_i})
\end{align*}
with
$$H_t(x,y,z,\tilde z, b,w):=y-F_t(e^x,y,z,\tilde z)\Delta t_i+zb+\tilde z w.$$
Inspired by \cite{hure2019some}, we design the numerical approximation of $u_{t_i}(X_{t_i})$ as follows:
\begin{enumerate}
\item[(1)] start with   $\widehat{\mathcal{U}}_N = G$;
\item[(2)] for $i=N-1,...,0$, given $\widehat{\mathcal{U}}_{i+1}$, use the triple of deep neural networks
\begin{align}
(\mathcal{U}_i(\cdot,\theta),\mathcal{Z}_i(\cdot,\theta),\tilde {\mathcal{Z}}_i(\cdot,\theta) )\in \,\,
&
\mathcal{N}\mathcal{N}^{\varrho}_{1+2i,1,M,m}(\mathbb{R}^{M_m})
\times\mathcal{N}\mathcal{N}^{\varrho}_{1+2i,1,M,m}(\mathbb{R}^{M_m})
\nonumber \\
&
\times\mathcal{N}\mathcal{N}^{\varrho}_{1+2i,1,M,m}(\mathbb{R}^{M_m})
\label{Neural-Network-i}
\end{align}
 for the approximation of
 $$
 \left(u_{t_i}(X_{t_i}),
 	\sqrt{(1-\rho^2)V_{t_i}} Du_{t_i}(X_{t_i} ),
 		\psi_{t_i}(X_{t_i} )+\rho\sqrt{V_{t_i}}Du_{t_i}(X_{t_i})\right),
 $$
 to achieve an estimate
 $$\mathcal{U}_{i+1}=H_{t_i}\left( X_{t_i},\mathcal{U}_i(X_{t_i},\theta_i),\mathcal{Z}_i(X_{t_i},\theta_i),\tilde{\mathcal{Z}}_i(X_{t_i},\theta_i),\Delta B_{t_i},\Delta W_{t_i}\right);$$
 \item[(3)]  compute the minimizer of the expected quadratic loss function
{\small
\begin{equation*}
\left\{
\begin{split}
\hat{L}_i(\theta): & = {E} \left|
	\widehat{\mathcal{U}}_{i+1}
	-H_{t_i}\left( X_{t_i},\mathcal{U}_i(X_{t_i},\theta_i),\mathcal{Z}_i(X_{t_i},\theta_i),\tilde{\mathcal{Z}}_i(X_{t_i},\theta_i),\Delta B_{t_i},\Delta W_{t_i}\right)
	\right|^2,\\
& \approx  \frac{1}{J}  \sum_{j=1}^J  \left|
	\widehat{\mathcal{U}}_{i+1}^{(j)}
	-H_{t_i}\left( X^{(j)}_{t_i},\mathcal{U}_i(X^{(j)}_{t_i},\theta_i),\mathcal{Z}_i(X^{(j)}_{t_i},\theta_i),\tilde{\mathcal{Z}}_i(X^{(j)}_{t_i},\theta_i),\Delta B^{(j)}_{t_i},\Delta W^{(j)}_{t_i}\right)
	\right|^2
\\
\theta^*_i & \in \arg\min\limits_{\theta\in\mathbb{R}^{M_m}}\hat{L}_i(\theta),
\end{split}
\right.
\end{equation*}
}
where the Adam (adaptive moment estimation) optimizer may be used to get the optimal parameter $\theta^*$;
\item[(4)] update and set $\widehat{\mathcal{U}}_i=\mathcal{U}_i(\cdot,\theta^*_i)$, $\widehat{\mathcal{Z}}_i=\mathcal{Z}_i(\cdot,\theta^*_i)$, and $\widehat{\tilde{\mathcal{Z}} }_i=\tilde{\mathcal{Z}}_i(\cdot,\theta^*_i)$.
\end{enumerate}
\begin{rmk}
Here, $(X^{(j)}, B^{(j)},W^{(j)}, \widehat W^{(j)}, V^{(j)})_{1\leq j \leq J}$ are independent simulations of $(X,B,W,\widehat W, V)$. Noticing that $\sF^W_t=\sF^{W,\widehat W}_t$ for $t\in [0,T]$, by Proposition \ref{prop-infty} and Remark \ref{rmk-filtration} we have the functions in $\mathcal{N}\mathcal{N}^{\varrho}_{1+2 i,1,M,m}(\mathbb{R}^{M_m})$ of the following form:
$$
\Phi_m(W_{t_1},\cdots, W_{t_{i}}, \widehat W_{t_1},\cdots, \widehat W_{t_{i}},x), \quad i=0,1,2,\cdots, N-1,
$$
which incorporates all the simulated values of $(W,\, \widehat W)$ until time $t_i$, leading to the changing dimension of the inputs. One may also see that the finer the partition of $[0,T]$ is, the higher input dimension it involves. The changing and high dimensionality arising from the approximations prompts us to adopt a deep learning-based method, and this also unveils  the difference from the scheme in \cite{hure2019some}.

On the other hand,  a convergence analysis of the above scheme is given in the appendix. Even though we are working with dimension-changing neural networks under a non-Markovian framework with different assumptions, we adopt a similar strategy to \cite{hure2019some} for the proof of the convergence analysis.
\end{rmk}


\subsection{Numerical examples for the rough Bergomi model}\label{section:numerical}


\subsubsection{European put option}

%

We consider the rough Bergomi model of~\cite{bayer2016pricing} in Example \ref{ex:rBergomi} with the
following choice of parameters: $H = 0.07$, $\eta = 1.9, \rho = -0.9$, $r=0.05$, $T=1$, $X_0=\ln (100)$. For
simplicity, we choose the forward variance curve to be $\xi(t) \equiv 0.09$,
independent of time.

We compute the numerical approximations to the European option price given in \eqref{European-optn}. The value function $u$ together with another random field $\psi$ constitutes the unique solution to BSPDE \eqref{BSPDE-value-funct} which corresponds to the BSPDE \eqref{BSPDE} in Theorem \ref{thm-existence} with
\begin{align*}
F_s(x,y,z,\tilde z) = -r y,\quad \text{and}\quad G(e^x)=(K- e^{x+rT})^+.
\end{align*}
By Theorems \ref{thm-uniqueness} and \ref{thm-existence},
the triple $(Y_t^{0,x},Z_t^{0,x},\overline Z_t^{0,x})_{t\in[0,T]}$ with
\begin{align*}
Y_t^{0,x}:=u_t(X_t^{0,x}),
\quad Z_t^{0,x}:=\rho \sqrt{V_t}Du_t(X_t^{0,x})+\psi_t(X_t^{0,x}),
\quad \overline Z_t^{0,x}:= \sqrt{(1-\rho^2)V_t} Du_t(X_t^{0,x}),
\end{align*}
for $t\in[0,T]$ satisfies the following FBSDE:
\begin{equation}\label{FBSDE-rB}
\left\{
\begin{split}
dX_s^{0,x}&= \sqrt{V_s} \left( \rho \,dW_s+ \sqrt{1-\rho^2} \,dB_s  \right) -\frac{V_s}{2}\,ds, \quad 0\leq s\leq T;\\
X_0^{0,x}&=x;\\
V_s&=\xi_s \, \cE(\eta\, \widehat W_s) \quad \text{with} \quad \widehat W_s= \int_0^s \sqrt{2H} (s-r)^{H-1/2}\,dW_r, \quad s\in[0,T];\\
dY_s^{0,x}&= rY_s^{0,x}\,ds+Z_s^{0,x}\,dW_s+\overline Z_s^{0,x}\, dB_s, \quad s\in[0,T];\\
Y_T^{0,x}&=G(e^{X_T^{0,x}}).
\end{split}\right.
\end{equation}

Then the deep learning-based method in Section \ref{Sec:deep-learning} is used for the numerical approximations.
We take $N=20$ in the Euler Scheme and set a single hidden layer whose number of neurons is equal to half of the total number of neurons in the input and output layers. We adopt the Sigmoid function for the activation function and the optimization algorithm is Adam. We implement 10000 trajectories in mini-batch and check the loss convergence every 50 iterations. In the following Table 1, the reference values are calculated by Monte Carlo method and they are close to the results obtained by averaging 20 independent runs with the deep learning method.
\begin{table}[H]
\centering
\begin{tabular}{c|c|c|c|c}
    & Reference value& $RSD=\frac{standard\ deviation}{average\ value}$ & Estimated value & RSD\\ \hline
    $K=90$& $4.9550$ & $0.0259$ & $4.9535$ & $0.0228$ \\
    $K=100$& $7.8284$ & $0.0135$ & $7.8061$ & $0.0201$\\
    $K=110$& $12.1844$ & $0.0100$ & $12.1940$ & $0.0143$ \\
    $K=120$& $18.1631$ & $0.0077$ & $18.1699$ & $0.0055$ \\
\end{tabular}
\caption{Prices of European put options at t=0 under the different strike prices K.}
\label{Table1}
\end{table}

On the other hand, we also investigate the dependence of the value function on the paths of process $V$. We simulate 10000 independent trajectories of the stochastic variance process $V$ and evaluate the corresponding values of $u(0.5,\ln 100)$ when $t=0.5$, $x=\ln 100$, and $K=100$. The mean of these $u(0.5,\ln 100)$ is $9.9287$ and the standard deviation $0.4240$. Four of these trajectories are randomly chosen in Figure 1 (a), and the corresponding values of $u(0.5,\ln 100)$ are listed in Table 2. From Figure 1(a) and Table 2, one may see that bigger values of $V(0.5)$ do not always lead to bigger option prices.
	Meanwhile, for the simulated 10000 trajectories of $V$, we reset the values of V to be the same and equal to the average of simulated values of $V(t)$  at time $t = 0.5$, i.e., we fix $V(0.5)=0.0825$. Then the mean of these values of $u(0.5,\ln 100)$ turns out to be $9.9292$ with the standard deviation equal to $0.4226$. Four of the trajectories corresponding to Figure 1 (a) are drawn in Figure 1 (b), and we show the corresponding values of $u(0.5,\ln 100)$ in Table 3. Comparing the obtained means, the standard deviations, and the four paths and associated values of $u(0.5,\ln 100)$ in these two cases,  we may see that the value of $V$ at $t=0.5$ does not play a dominating role in determining the price of the options $u(0.5,\ln 100)$, which is different from the classical Markovian cases;  this is due to the path-dependence and thus the non-Markovianity, i.e., the trajectory of $V$ before $t=0.5$ actually affects the value of $u(0.5,\ln 100)$ in a non-negligible manner.\\
\begin{figure}[H]
 \centering
 \subfigure[Paths of V with different values at t=0.5]{
 \includegraphics[width=0.45\textwidth]{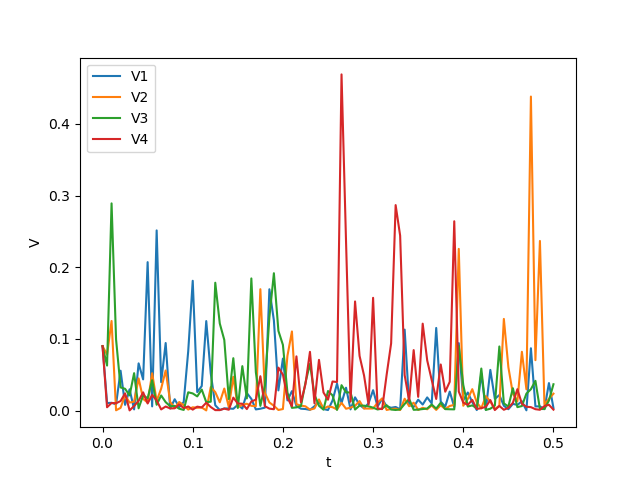}}
 \hspace{0in}
 \subfigure[Paths of $V$ with a fixed value at $t=0.5$]{
 \includegraphics[width=0.45\textwidth]{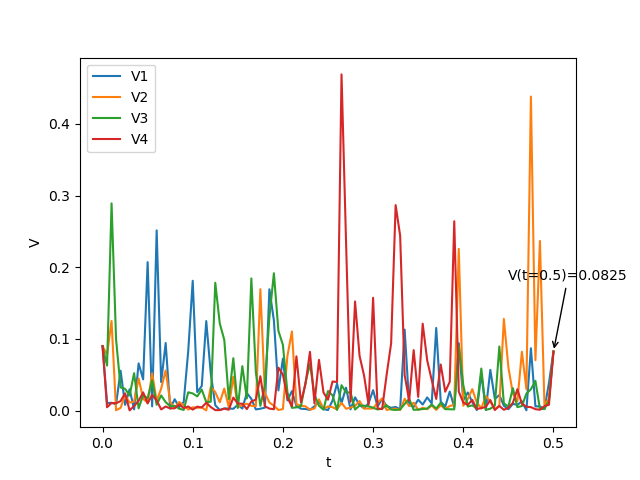}}
 \caption{Different paths of $V$ on time interval $[0,0.5]$}
 \label{Fig1}
\end{figure}
\begin{table}[H]
\centering
\begin{tabular}{c|c|c}
    \text{Paths of process V with } & $u(0.5,\ln 100)$\\ \hline
    $V1(t=0.5)=0.0038$  & $10.3310$ \\
    $V2(t=0.5)=0.0237$  & $10.4847$ \\
    $V3(t=0.5)=0.0369$  & $10.1519$ \\
    $V4(t=0.5)=0.0014$  & $10.3003$ \\
\end{tabular}
\caption{$u(0.5,\ln 100)$ on different paths of $V$ in Figure 1(a)}
\label{Table2}
\end{table}

\begin{table}[H]
\centering
\begin{tabular}{c|c|c}
    \text{Paths of process V with } & $u(0.5,\ln 100) $\\ \hline
    $V1(t=0.5)=0.0825$  & $10.3290$ \\
    $V2(t=0.5)=0.0825$  & $10.4878$ \\
    $V3(t=0.5)=0.0825$  & $10.1457$ \\
    $V4(t=0.5)=0.0825$  & $10.2769$ \\
\end{tabular}
\caption{$u(0.5,\ln 100)$ on different paths of $V$ in Figure 1(b)}
\label{Table3}
\end{table}

\subsubsection{Two schemes for approximating American put options}
 Again, consider the rough Bergomi model in Example \ref{ex:rBergomi} with the
following choice of parameters: $H = 0.07$, $\eta = 1.9, \rho = -0.9$, $r=0.05$, $T=1$, $X_0= \ln (100) $. Also, we choose the forward variance curve to be $\xi(t) \equiv 0.09$ independent of time,  for simplicity. The strike prices may take different values. Then, pricing the American put option is to compute
\begin{align*}
\overline u_0(x):= \sup_{\tau\in \mathcal T_0}E\left[   e^{-\tau r}g_{\tau}(e^{X_{\tau}^{0,x} })  \right],  \text{ with }  g_{\tau}(e^x)=(K-e^{r\tau+x})^+, \text{ for } (\tau,x)\in[0,T]\times \bR.
\end{align*}
We shall adopt two different schemes for the computations for the numerical approximations.

The first scheme is based on the penalization. By Corollary \ref{cor-AO}, $\overline u_0(x)$ may be approximated by $\overline u^{\tilde{N}}_0(x)$ as $\tilde{N}$ tends to infinity, where the pair $(\overline u ^{\tilde{N}}, \overline \psi ^{\tilde{N}})$ is the unique weak solution to
BSPDE \eqref{BSPDE} with
\begin{align*}
F_t(e^x,y,z,\tilde z)
=-ry+\tilde N \left(   g_t(e^x) - y  \right)^+ \quad \text{and}\quad
G(e^x)=g_T(e^x).
\end{align*}
Then the first scheme is to use the algorithm in Section \ref{Sec:deep-learning} to compute $\overline u_0^{\tilde N}(X_0)$ which approximates $\overline u_0(x)$ when $\tilde N$ tends to infinity.

The second scheme is based on the representation via the following forward-backward system:
\begin{equation}\label{R-FBSDE-rB}
\left\{
\begin{split}
dX_s^{0,x}&= \sqrt{V_s} \left( \rho \,dW_s+ \sqrt{1-\rho^2} \,dB_s  \right) -\frac{V_s}{2}\,ds, \quad 0\leq s\leq T;\\
X_0^{0,x}&=x;\\
V_s&=\xi_s \, \cE(\eta\, \widehat W_s) \quad \text{with} \quad \widehat W_s= \int_0^s \sqrt{2H} (s-r)^{H-1/2}\,dW_r, \quad s\in[0,T];\\
-d\overline Y_s^{0,x}&=-r\overline Y_s^{0,x}\,ds +dA^{0,x}_s-\overline Z_s^{0,x;B}\,dB_s -\overline Z_s^{0,x;W}\,dW_s,\quad s\in[0,T];\\
\overline Y_T^{0,x}&=g_T(e^{X_T^{0,x}})   ;\quad \overline Y_s^{0,x}\geq       g_s(e^{X_s^{0,x}}),\quad s\in [0,T];\\
A^{0,x}_{\cdot}& \text{ is increasing and continuous, }A_0^{0,x}=0, \int_0^T \left(\overline Y_s^{0,x}-  g_s(e^{X_s^{0,x}}) \right)\,d A_s^{0,x}=0.
\end{split}\right.
\end{equation}
Recalling the assertion (v) in Corollary \eqref{cor-AO} which gives the following representation
\begin{align*}
\overline u_{\tau}(X_{\tau}^{0,x})=\overline Y_{\tau}^{0,x},\quad
\overline \psi^B_{\tau}(X_{\tau}^{0,x}) =\overline Z_{\tau}^{0,x;B}, \quad \text{and}\quad
\overline \psi^W_{\tau}(X_{\tau}^{0,x})
= \overline Z_{\tau}^{0,x;W},\quad \text{a.s.},
\end{align*}
for $0\leq \tau \leq T$, for some triple $(\overline u,\,\overline \psi^B,\overline \psi^W)$ defined on $(\Omega\times [0,T]\times \bR,\sP^W\otimes \cB(\bR))$, we may use the following scheme:
 \begin{enumerate}
 \item [(1)]  Start with $\widehat{\mathcal{U}}_N = g_T$.
 \item [(2)] For $i=N-1,...,0$, given $\widehat{\mathcal{U}}_{i+1}$, use the triple of deep neural networks
\begin{align}
(\mathcal{U}_i(\cdot,\theta),\mathcal{Z}^B_i(\cdot,\theta),\tilde {\mathcal{Z}}^W_i(\cdot,\theta) )\in \,\,
&
\mathcal{N}\mathcal{N}^{\varrho}_{1+2i,1,M,m}(\mathbb{R}^{M_m})
\times\mathcal{N}\mathcal{N}^{\varrho}_{1+2i,1,M,m}(\mathbb{R}^{M_m})
\nonumber\\
&
\times\mathcal{N}\mathcal{N}^{\varrho}_{1+2i,1,M,m}(\mathbb{R}^{M_m}) \label{neural-network-AMO}
\end{align}
 for the approximation of
 $
 \left(\overline u_{t_i}(X_{t_i}),
 	\overline\psi^B_{t_i}(X_{t_i} ),
 		\overline \psi^W_{t_i}(X_{t_i})\right),
 $
 and obtain an estimate
 $$\mathcal{U}_{i+1}= \mathcal{U}_i(X_{t_i},\theta_i) + r \mathcal{U}_i(X_{t_i},\theta_i) \Delta t_i
 + \mathcal{Z}^B_i(X_{t_i},\theta_i)\, \Delta B_{t_i}+ \mathcal{Z}^W_i(X_{t_i},\theta_i)\,\Delta W_{t_i}.$$
 \item [(3)] Compute the minimizer of the expected quadratic loss function:
\begin{equation*}
\left\{
\begin{split}
\hat{L}_i(\theta): & = {E} \left|
	\widehat{\mathcal{U}}_{i+1}
	-
	\mathcal{U}_{i+1}
	\right|^2,\\
\theta^*_i & \in \arg\min\limits_{\theta\in\mathbb{R}^{N_m}}\hat{L}_i(\theta).
\end{split}
\right.
\end{equation*}
\item [(4)] Update $\widehat{\mathcal{U}}_i= \max\left\{ \mathcal{U}_i(X_{t_i},\theta^*_i),\, g_{t_i}(X_{t_i}) \right\}$.
 \end{enumerate}

The above scheme extends the one proposed in \cite[Section 3.3]{hure2019some} from Markovian cases to a non-Markovian setting, with the main difference lying in the changing dimensions in the neural networks \eqref{neural-network-AMO}. Looking into Appendix for the convergence analysis of the scheme in Section \ref{Sec:deep-learning}, we may extend the  convergence analysis in  \cite[Section 4.3]{hure2019some} to our non-Markovian setting, and as such an extension is similar to that of the scheme in Section \ref{Sec:deep-learning}, the proof is omitted.

In Table 4, the estimates of the above two schemes are presented together with the reference values which are lower bound estimates from \cite{bayer2018pricing}. We take $N=20$ and implement a single hidden layer whose number of neurons is equal to half of the total number of neurons in the input and output layers. The activation function and optimization algorithm we use here are Sigmoid function and Adam. The results are obtained by averaging 20 independent runs. For the first scheme, in theory, $\overline u_0^{\tilde N}(X_0)$ is (bigger and) closer to the real value than $\overline u_0^{\bar N}(X_0)$ when $\tilde N > \bar N$, which is affirmed by the numerical experiments.
We set $\tilde N$ equal to 40 and 10000 for comparisons. The same neural networks are put to use in the second scheme. Here, neural networks with $\geq 2$ hidden layers and/or big number of neurons were also tried, which, we believed, might produce better approximations. However, we found the obtained results were largely different and quite sensitive to the learning rate,  the optimizer,  the iteration numbers, and even the activation function, and this enlightened us to reduce the complexity to use the selected neural networks for relatively stable estimates.   
\begin{table}[H]
\centering
\begin{tabular}{c|c|c|c|c|c|c|c}
\hline
\multirow{2}{*}{ }&
\multirow{2}{*}{reference value} &
\multicolumn{4}{c|}{1st scheme} &
\multirow{2}{*}{2nd scheme} &
\multirow{2}{*}{RSD} \\
\cline{3-6}
  & & N=40 & RSD & N=10000& RSD &{}\\
\hline
$K=90$ & $5.32$ & $5.5053$ & $0.0998$ & $5.5113$ & $0.0980$ &$5.5497$ & $0.0895$ \\
$K=100$ & $8.51$ & $9.6392$ & $0.0553$ & $9.6672$ & $0.0582$ & $9.6867$ & $0.0552$ \\
$K=110$ & $13.24$ & $15.4707$ & $0.0196$ & $15.4882$ & $0.0243$ & $15.5020$ & $0.0292$ \\
$K=120$ & $20$ & $22.5800$ & $0.0213$ & $22.6069$ & $0.0221$ &$22.5742$ & $0.0114$ \\
\end{tabular}
\caption{Prices of American put options at t=0 under two different schemes.}
\label{Table4}
\end{table}

While the two schemes presented in this paper yield results that are very
close to each other (well within confidence intervals for the Monte Carlo
error), the references values from \cite{bayer2018pricing} differ
significantly. It should be noted that the results from
\cite{bayer2018pricing} -- which were also recovered by a similar method
suggested in \cite{goudenege2020machine} -- are only supported by theory for
Markov models. Moreover, those results are lower bounds, and currently, to the
best of our knowledge, no efficient numerical methods providing upper bounds
of American option prices in rough volatility models has been provided. In
contrast, our method is supported by theory. In essence, this leads us to the
uncomfortable conclusion that either the reference values from
\cite{bayer2018pricing} or our own results -- or both -- are highly
inaccurate, and that we are unable to discern which.

In order to backtest our algorithm, we additionally consider a
classical Markovian case, setting  $\rho=\eta=0$ and keeping the other parameters unchanged,
The estimates of the above two schemes are compared with the option prices calculated by \textit{binprice} function in the financial toolbox of Matlab. It can be seen from Table 5 that our results are pretty close to the option price estimates by using the Cox-Ross-Rubinstein binomial model.
\begin{table}[H]
\centering
\begin{tabular}{c|c|c|c|c|c|c|c}
\hline
\multirow{2}{*}{ }&
\multirow{2}{*}{Reference value} &
\multicolumn{4}{c|}{1st scheme} &
\multirow{2}{*}{2nd scheme} &
\multirow{2}{*}{RSD} \\
\cline{3-6}
  & & N=40 & RSD & N=10000& RSD &{}\\
\hline
    $K=90$& $5.6168$ & $5.5700$ & $0.0949$ & $5.5945$ & $0.0931$ &$5.6157$ & $0.0881$ \\
    $K=100$& $9.7980$ & $9.7465$ & $0.0520$ & $9.7779$ & $0.0504$ &$9.7928$ & $0.0555$\\
    $K=110$& $15.6720$ & $15.6176$ & $0.0265$ &$15.6516$ & $0.0210$ & $15.6341$ & $0.0221$\\
    $K=120$& $22.7501$ & $22.7140$ & $0.0204$ & $22.7367$ & $0.0185$ & $22.6994$ & $0.0106$\\
\end{tabular}
\caption{American option prices when $\rho=\eta=0$.}
\label{Table5}
\end{table}




\appendix

\section{Convergence analysis}
\label{sec:convergence-analysis}

This section is to devoted to a convergence analysis for the deep learning-based scheme proposed in Section \ref{Sec:deep-learning}. The discussions are conducted under Assumptions $(\cA^*)$, \ref{ass:stoch-var}, \ref{ass}, and the following one:\\[5pt]
{\bfseries (H1)}  (i) There exists a continuous and increasing function $\rho: [0,\infty) \rightarrow [0,\infty)$ with $\rho(0)=0$ such that for any $0\leq t_1\leq t_2\leq  T$, it holds that
$$
E\left[ \int_{t_1}^{t_2} V_s\,ds \right]  + E\left[\left(\int_{t_1}^{t_2} V_s\,ds\right)^2\right]\leq \rho(|t_1-t_2|).
$$
(ii) There exists a constant $L_2 > 0$ such that
\begin{align*}
&|F_{t_1}( e^{x_1}, y_1, z_1,\tilde{z}_1)-F_{t_2}( e^{x_2}, y_2, z_2,\tilde{z}_2)|
\\
&\leq L_2 (\sqrt{\rho(|t_2-t_1|)}+|x_2-x_1|+|y_2-y_1|+|z_2-z_1|+|\tilde{z}_2-\tilde{z}_1|),
\end{align*}
for all $(t_1, x_1, y_1, z_1,\tilde{z}_1)$ and $(t_2, x_2, y_2, z_2,\tilde{z}_2)$ in $[0, T] \times \mathbb{R} \times \mathbb{R} \times \mathbb{R} \times \mathbb{R}.$

\begin{rmk}
In fact, for examples like \ref{ex:rBergomi} and \ref{ex:rough-Heston}, one has
$$
E\left[ \int_0^T |V_t|^p\, dt \right] <\infty, \quad \text{for some }p>2,
$$
which, by H\"older's inequality, implies
\begin{align*}
&E\int_{t_1}^{t_2} V_t \,dt \leq |t_1-t_2|^{\frac{p-1}{p}} \left( E\int_{t_1}^{t_2} |V_t|^p\,dt \right)^{1/p} \leq C_p \left|t_1-t_2\right|^{\frac{p-1}{p}},\quad \text{for }0\leq t_1\leq t_2\leq T,\\
&E\left[\left(\int_{t_1}^{t_2} V_t \,dt\right)^2\right]
\leq |t_1-t_2|^{\frac{2(p-1)}{p}} \left( E\int_{t_1}^{t_2} |V_t|^p\,dt \right)^{2/p} \leq C^2_p \left|t_1-t_2\right|^{\frac{2(p-1)}{p}},\quad \text{for }0\leq t_1\leq t_2\leq T,
\end{align*}
and thus, we may take $\rho(r)=(C_p+C_p^2) \cdot \left( |r|^{\frac{p-1}{p}} \vee  |r|^{\frac{2(p-1)}{p}}\right)$, for $r\geq 0$. Further, one may straightforwardly check that the numerical examples discussed in Section \ref{section:numerical} have Assumption {\bfseries (H1)} satisfied.
\end{rmk}

In what follows, we denote by $C$ a positive generic constant whose value is independent of $\pi$ and may vary from line to line, and by $\mathcal{X}$ we denote the unique (strong) solution to the SDE \eqref{forward-X} start at $t=0$ and by $X = X^{\pi}$ the Euler-Maruyama approximation with a time grid $\pi = \{t_0 = 0 < t_1 < ... < t_N = T\}$, with modulus $|\pi| =\max_{1\leq i\leq N}|t_i-t_{i-1}|$ bounded by  $\frac{CT}{N}$ for some constant $C$. Under Assumptions \ref{ass:stoch-var} and \textbf{(H1)}, standard calculations yield that
\begin{align}
 {E}\left[\sup_{0\le t\le T} |\mathcal{X}_t|^2\right]  \le C (1+|x_0|^2)\label{X},\\
\max_{i=0,...,N-1} {E} \left[|\mathcal{X}_{t_{i+1}}-X_{t_{i+1}}|^2 + \sup_{t\in[t_i,t_{i+1}]}|\mathcal{X}_t-X_{t_i}|^2\right]\le C\rho(|\pi|). \label{rho}
\end{align}
By the theory of BSDEs (see \cite{Hu_2002} for instance),  Assumptions \ref{ass:stoch-var}, \ref{ass}, and  {\bfseries (H1)} imply the existence and uniqueness of an adapted $L^2$-solution $(Y,Z,\tilde{Z})$ to BSDE \eqref{BSDE}, which together with \eqref{X} and {\bfseries (H1)}-(ii) gives
\begin{equation}\label{F}
 {E}\left[\int_{0}^{T}|F_t(e^{\mathcal{X}_t},Y_t,Z_t,\tilde{Z}_t)|^2dt\right]<\infty
\end{equation}
and the standard $L^2$-regularity result on $Y$:
\begin{equation}\label{Y}
\max_{i=0,...,N-1} {E}\left[\sup_{t\in[t_i,t_{i+1}]}|Y_t-Y_{t_i}|^2 \right] = O(|\pi|).
\end{equation}

For the pair $(Z,\tilde Z)$, set
\begin{equation}
\left\{\begin{array}{l}{
 \varepsilon^{Z} (\pi) : = {E}\left[\sum_{i=0}^{N-1}\int_{t_i}^{t_{i+1}}|Z_t-\bar{Z}_{t_i}|^2dt\right],\quad \text{with } \bar{Z}_{t_i} := \frac{1}{\Delta t_i} {E}_i\left[\int_{t_i}^{t_{i+1}}Z_tdt\right]},\\
{\varepsilon^{\tilde{Z}}(\pi) :=  {E}\left[ \sum_{i=0}^{N-1}\int_{t_i}^{t_{i+1}}|\tilde{Z}_t-\bar{\tilde{Z}}_{t_i}|^2dt\right],\quad \text{with } \bar{\tilde{Z}}_{t_i} := \frac{1}{\Delta t_i} {E}_i\left[\int_{t_i}^{t_{i+1}}\tilde{Z}_tdt\right] },
\end{array}\right.
\end{equation}
where $ E _i$ denotes the conditional expectation given $\mathcal{F}_{t_i}$. 

To investigate the convergence of the deep learning scheme, we define, for $i = 0,...,N-1$,
\begin{equation}\label{VZZ}
\left\{\begin{array}{l}
{\widehat{\mathcal{V}}_{t_{i}} := E _{i}[\widehat{\mathcal{U}}_{i+1}(X_{t_{i+1}})]+F_{t_{i}}( e^{X_{t_{i}}}, \widehat{\mathcal{V}}_{t_{i}}, \overline{\widehat{Z}_{t_{i}}}, \overline{\widehat{\tilde{Z}}_{t_{i}}})\Delta t_{i},} \\ {\overline{\widehat{Z}_{t_{i}}} :=\frac{1}{\Delta t_{i}}  E _{i}[(\widehat{\mathcal{U}}_{i+1}(X_{t_{i+1}})\Delta B_{t_{i}}],} \\ {\overline{\widehat{\tilde{Z}}_{t_{i}}} :=\frac{1}{\Delta t_{i}}  E _{i}[(\widehat{\mathcal{U}}_{i+1}(X_{t_{i+1}})\Delta W_{t_{i}}],}
\end{array}\right.
\end{equation}
where, $\widehat{\mathcal{V}}_{t_i}$ is well-defined for sufficiently small $|\pi|$ due to the uniform Lipschitz continuity of $F$. In view of Theorem \ref{thm-uniqueness}, we may find $\sF^W_{t_i}\otimes \cB(\bR)$-measurable functions $\hat{v}_i$, $\overline{\hat{z}_i}$, and $\overline{\hat{\tilde{z}}_i}$ s.t.
\begin{equation}\label{4.6}
\widehat{\mathcal{V}}_{t_i} = \hat{v}_i(X_{t_i}),\quad \overline{\widehat{Z}_{t_i}} = \overline{\hat{z}_i}(X_{t_i}),\quad and \quad \overline{\widehat{\tilde{Z}}_{t_i}} = \overline{\hat{\tilde{z}}_i}(X_{t_i}), \quad i=0,...,N-1.
\end{equation}
On the other hand, by the martingale representation theorem, there exist two $\mathbb{R}$-valued square integrable processes $\{\widehat{Z}_t\}$ and $\{\widehat{\tilde{Z}}_t\}$ s.t.
\begin{equation}\label{4.7}
\widehat{\mathcal{U}}_{i+1}(X_{t_{i+1}})=\widehat{\mathcal{V}}_{t_i}-F_{t_i}(e^{X_{t_i}},\widehat{\mathcal{V}}_{t_i},\overline{\widehat{Z}_{t_i}},\overline{\widehat{\tilde{Z}}_{t_i}})\Delta t_i +\int_{t_i}^{t_{i+1}}\widehat{Z}_t\, dB_t+\int_{t_i}^{t_{i+1}}\widehat{\tilde{Z}}_t\, dW_t,
\end{equation}
and It\^{o}'s isometry gives
\begin{center}
$\overline{\widehat{Z}_{t_i}} = \frac{1}{\Delta t_i} E _i[\int_{t_i}^{t_{i+1}}\widehat{Z}_tdt],\quad \overline{\widehat{\tilde{Z}}_{t_i}} = \frac{1}{\Delta t_i} E _i[\int_{t_i}^{t_{i+1}}\widehat{\tilde{Z}}_tdt],\quad i=0,...,N-1.$
\end{center}
The distance between the optimal triple $(\widehat{\mathcal{U}}_{i}, \, \widehat{\mathcal{Z}}_{i},\,\widehat{\mathcal{\tilde{Z}}}_{i} )$ from the deep learning-based scheme and $({\widehat{\mathcal{V}}_{t_{i}}}, \, {\overline{\widehat{Z}_{t_{i}}}},\, {\overline{\widehat{\tilde{Z}}_{t_{i}}}} )$ from the system \eqref{VZZ} is given as follows.
\begin{lem}\label{lem-dis-v}
Let Assumptions $(\cA^*)$, \ref{ass:stoch-var}, \ref{ass}, and {\bfseries (H1)} hold.  When $|\pi|$  is sufficiently small, we have
\begin{align}
 &E |\widehat{\mathcal{V}}_{t_{i}}-\widehat{\mathcal{U}}_{i}(X_{t_{i}})|^{2}
 +\Delta t_{i}  E\left[ | \overline{\widehat{Z}_{t_{i}}}-\widehat{\mathcal{Z}}_i(X_{t_{i}})|^{2}
 +  | \overline{\widehat{\tilde{Z}}_{t_{i}}}-\widehat{\tilde{\mathcal{Z}}}_i(X_{t_{i}})|^{2} \right]
 \nonumber\\
 &\leq C \varepsilon_{i}^{\mathcal{N}, v}+C \Delta t_{i} \varepsilon_{i}^{\mathcal{N}, z}+C \Delta t_{i} \varepsilon_{i}^{\mathcal{N}, \tilde{z}},\label{4.18}
\end{align}
where we use
\begin{align*}
\varepsilon_{i}^{\mathcal{N}, v} :=\inf _{\xi}  E |\hat{v}_{i}(X_{t_{i}})-\mathcal{U}_{i}(X_{t_{i}} ; \xi)|^{2},
\quad \varepsilon_{i}^{\mathcal{N}, z} :=\inf _{\eta}  E |\overline{\hat{z}_i}(X_{t_{i}})-\mathcal{Z}_{i}(X_{t_{i}} ; \eta)|^{2},
\end{align*}
and $ \varepsilon_{i}^{\mathcal{N}, \tilde{z}} :=\inf _{\eta}  E |\overline{\hat{\tilde{z}}_i}(X_{t_{i}})-\tilde{\mathcal{Z}}_{i}(X_{t_{i}} ; \eta)|^{2}$ to denote the $L^2$-approximation errors of $\hat{v}_i$ ,$\overline{\hat{z}_i}$, and $\overline{\hat{\tilde{z}}_i}$ by neural networks $\mathcal{U}_i$ , $\mathcal{Z}_i$, and $\tilde{\mathcal{Z}}_i$,  for $i = 0,...,N-1$.
\end{lem}
To focus on the convergence analysis, we postpone the proof of Lemma \ref{lem-dis-v}. Define the following square error:
\begin{align*}
\mathcal{E}[(\widehat{\mathcal{U}}, \widehat{\mathcal{Z}}, \widehat{\mathcal{\tilde{Z}}}),(Y, Z, \tilde{Z})]
&=\max _{i=0, \ldots, N-1}  E \left[ |Y_{t_{i}}-\widehat{\mathcal{U}}_{i}(X_{t_{i}})|^{2} \right]+ E \left[\sum_{i=0}^{N-1} \int_{t_{i}}^{t_{i+1}}|Z_{t}-\widehat{\mathcal{Z}}_{i}(X_{t_{i}})|^{2}dt\right]
\\
&\quad\quad+ E \left[\sum_{i=0}^{N-1} \int_{t_{i}}^{t_{i+1}}|\tilde{Z}_{t}-\widehat{\mathcal{\tilde{Z}}}_{i}(X_{t_{i}})|^{2}dt\right].
\end{align*}


\begin{thm}\label{thm-converg}
 Under Assumptions $(\cA^*)$,  \ref{ass:stoch-var}, \ref{ass}, and {\bfseries (H1)}, it holds that
{\small
\begin{align}
&\mathcal{E}[(\widehat{\mathcal{U}}, \widehat{\mathcal{Z}}, \widehat{\tilde{\mathcal{Z}}}),(Y, Z, \tilde{Z})] \nonumber \\
& \leq C\left\{  E |G(\mathcal{X}_{T})-G(X_{T})|^{2}+\rho(|\pi|) + |\pi|+\varepsilon^{Z}(\pi)+\varepsilon^{\tilde{Z}}(\pi)+\sum_{i=0}^{N-1}(N \varepsilon_{i}^{\mathcal{N}, v}+\varepsilon_{i}^{\mathcal{N}, z}+\varepsilon_{i}^{\mathcal{N}, \tilde{z}})\right\}, \label{4.8}
\end{align}
}
where the constant $C$ is independent of the partition $\pi$.
\end{thm}
The computations involved in the proofs of Lemma \ref{lem-dis-v} and Theorem \ref{thm-converg} are conducted in a similar way to \cite[Section 4.1]{hure2019some} by  Hur{\'e}, Pham, and Warin, with the main differences lying in the approximations of the random variables with dimension-varying neural networks and the general modulus function $\rho(\pi)$. We provide the proofs for the reader's interests.
\begin{proof}[Proof of Theorem \ref{thm-converg}]
\textbf{Step 1.}
  We first derive a recursive estimate for the square norm of $Y_{t_{i}}-\widehat{\mathcal{V}}_{t_{i}}$, i.e.,  \begin{align}
   E |Y_{t_{i}}-\widehat{\mathcal{V}}_{t_{i}}|^{2} &\leq (1+C|\pi|)  E |Y_{t_{i+1}}-\widehat{\mathcal{U}}_{i+1}(X_{t_{i+1}})|^{2}
   +C|\pi|  E \left[\int_{t_{i}}^{t_{i+1}}|F_t(e^{\mathcal{X}_{t}}, Y_{t}, Z_{t},\tilde{Z}_t)|^{2}dt\right]
   \notag\\
   &
   +C  E \left[\int_{t_{i}}^{t_{i+1}}\left(|\tilde{Z}_{t}-\bar{\tilde{Z}}_{t_{i}}|^{2}+|Z_{t}-\bar{Z}_{t_{i}}|^{2} \right)dt\right]
   +C\rho(|\pi|)|\pi|, \label{4.12}
  \end{align}
  for each $i\in \{0,...,N-1\}$.

 In view of \eqref{BSDE} and \eqref{VZZ}, we have
\begin{align}
Y_{t_{i}}-\widehat{\mathcal{V}}_{t_{i}}=& E _{i}[Y_{t_{i+1}}-\widehat{\mathcal{U}}_{i+1}(X_{t_{i+1}})]+ E _{i}\left[\int_{t_{i}}^{t_{i+1}}F_t( e^{\mathcal{X}_{t}}, Y_{t}, Z_{t}, \tilde{Z}_t)-F_{t_{i}}( e^{X_{t_{i}}},\widehat{\mathcal{V}}_{t_{i}},\overline{\widehat{Z}_{t_{i}}},\overline{\widehat{\tilde{Z}}_{t_{i}}})\mathrm{d}t\right].\nonumber
\end{align}
 Young's inequality gives $(a + b)^2 \le (1 + \gamma\Delta t_i)a^2 + (1 + \frac{1}{\gamma\Delta t_i})b^2$  for any $a,b\in\bR$ and $\gamma>0$, which combined with Cauchy-Schwarz inequality, the Lipschitz condition on $F$ in {\bfseries (H1)}, and the estimation \eqref{rho} on the forward process, implies that
 \begin{align}\label{4.9}
 &E |Y_{t_i}-\widehat{\mathcal{V}}_{t_i}|^2
 \nonumber\\
 &\le
 E \bigg\{(1+\gamma\Delta t_i)\left( E _i[Y_{t_{i+1}}-\widehat{\mathcal{U}}_{i+1}(X_{t_{i+1}})]\right)^2\nonumber\\ &
 	\quad+\left(1+\frac{1}{\gamma\Delta t_i}\right)
	\left( E _i\Big[\int_{t_i}^{t_{i+1}}(F_t(e^{\mathcal{X}_t},Y_t,Z_t,\tilde{Z}_t)- F_{t_i}(e^{X_{t_i}},\widehat{\mathcal{V}}_{t_i},	
			\overline{\widehat{Z}_{t_i}},\overline{\widehat{\tilde{Z}}_{t_i}}))dt\Big]\right)^2\bigg\}
\nonumber\\
&\le
 	(1+\gamma\Delta t_i)
	E\left[ | E _i[Y_{t_{i+1}}-\widehat{\mathcal{U}}_{i+1}(X_{t_{i+1}})]|^2\right]
		+5\left(1+\frac{1}{\gamma\Delta t_i}\right)  L_2 ^2\Delta t_i\bigg\{C\rho(|\pi|)|\pi|
\nonumber\\
&\quad \quad
	+E \bigg[\int_{t_i}^{t_{i+1}}|Y_t-\widehat{\mathcal{V}}_{t_i}|^2dt\bigg]
	+E\bigg[\int_{t_i}^{t_{i+1}}\left( |Z_t-\overline{\widehat{Z}_{t_i}}|^2
	+|\tilde{Z}_t-\overline{\widehat{\tilde{Z}}_{t_i}}|^2 \right)dt\bigg]\bigg\}
	\nonumber\\
&\le
\left(1+\gamma\Delta t_i\right) E \left[ | E _i[Y_{t_{i+1}}-\widehat{\mathcal{U}}_{i+1}(X_{t_{i+1}})]|^2\right]
	+5\left(1+\gamma\Delta t_i\right)
		\frac{ L_2 ^2}{\gamma} \bigg\{C\rho(|\pi|)|\pi|
\nonumber\\
&\quad
+2\Delta t_i  E |Y_{t_i}
	-\widehat{\mathcal{V}}_{t_i}|^2
	+ E \bigg[ \int_{t_i}^{t_{i+1}}
	\left(|Z_t-\overline{\widehat{Z}_{t_i}}|^2 + |\tilde{Z}_t-\overline{\widehat{\tilde{Z}}_{t_i}}|^2 \right)dt
		\bigg] \bigg\},
\end{align}
where the $L^2$-regularity of $Y$ \eqref{Y} is used in the last inequality.

Recalling that $\bar{Z}$ and $\bar{\tilde{Z}}$  are the $L^2$-projections of $Z$ and $\tilde{Z}$ respectively,  we have\begin{equation}\label{4.10}
\left\{\begin{array}{l}{
 E [\int_{t_{i}}^{t_{i+1}}|Z_{t}-\overline{\widehat{Z}_{t_{i}}}|^{2}dt]
 	= E [\int_{t_{i}}^{t_{i+1}}|Z_{t}-\bar{Z}_{t_{i}}|^{2}dt]+\Delta t_{i}  E \left[ |\bar{Z}_{t_{i}}-\overline{\widehat{Z}_{t_{i}}}|^{2}\right] },\\
{ E [\int_{t_{i}}^{t_{i+1}}|\tilde{Z}_{t}-\overline{\widehat{\tilde{Z}}_{t_{i}}}|^{2}dt]
= E [\int_{t_{i}}^{t_{i+1}}|\tilde{Z}_{t}-\bar{\tilde{Z}}_{t_{i}}|^{2}dt]+\Delta t_{i}  E \left[|\bar{\tilde{Z}}_{t_{i}}-\overline{\widehat{\tilde{Z}}_{t_{i}}}|^{2}\right]}.
\end{array}\right.
\end{equation}
Integrate equation \eqref{BSDE} over time interval $[t_i,t_{i+1}]$ multiplied by $\Delta W_{t_i}$ and $\Delta B_{t_i}$ respectively.  This together with \eqref{VZZ} gives
\begin{align}
\Delta t_{i}\left(\bar{\tilde{Z}}_{t_{i}}-\overline{\widehat{\tilde{Z}}_{t_{i}}}\right)
=& E _{i}\left[
\Delta W_{t_{i}}\left(Y_{t_{i+1}}-\widehat{\mathcal{U}}_{i+1}(X_{t_{i+1}})- E _{i}[Y_{t_{i+1}}-\widehat{\mathcal{U}}_{i+1}(X_{t_{i+1}})]\right)\right]\nonumber\\
& +E _{i} \left[\Delta W_{t_{i}} \int_{t_{i}}^{t_{i+1}} F_t(e^{ \mathcal{X}_{t}} , Y_{t}, Z_{t}, \tilde{Z}_t)dt\right],\nonumber\\
\Delta t_{i}\left(\bar{Z}_{t_{i}}-\overline{\widehat{Z}_{t_{i}}}\right)
=& E _{i}\left[\Delta B_{t_{i}}\left(Y_{t_{i+1}}-\widehat{\mathcal{U}}_{i+1}(X_{t_{i+1}})- E _{i}[Y_{t_{i+1}}-\widehat{\mathcal{U}}_{i+1}(X_{t_{i+1}})]\right)\right]
\nonumber\\
&
	+E _{i}\left[ \Delta B_{t_{i}} \int_{t_{i}}^{t_{i+1}} F_t(e^{ \mathcal{X}_{t}}, Y_{t}, Z_{t}, \tilde{Z}_t)dt
	\right].\nonumber
\end{align}
Standard computations further indicate that
\begin{align}\label{4.11}
\Delta t_{i}  E \left[  |\bar{Z}_{t_{i}}-\overline{\widehat{Z}_{t_{i}}}|^{2}  \right]
\leq
& 2\left( E |Y_{t_{i+1}}-\widehat{\mathcal{U}}_{i+1}(X_{t_{i+1}})|^{2}- E | E _{i}[Y_{t_{i+1}}-\widehat{\mathcal{U}}_{i+1}(X_{t_{i+1}})]|^{2}\right)
\nonumber\\
&+2\Delta t_{i} E \left[\int_{t_{i}}^{t_{i+1}}|F_t(e^{ \mathcal{X}_{t}}, Y_{t}, Z_{t}, \tilde{Z}_{t})|^{2}dt\right];
\end{align}
it follows similarly for $\tilde{Z}$. Then, by plugging \eqref{4.10} and \eqref{4.11} into \eqref{4.9}, and choosing $\gamma = 20 L_2 ^2$, we have
\begin{align}
 &
 E \left[ |Y_{t_i}-\widehat{\mathcal{V}}_{t_i}|^2 \right]
 \notag \\
 &\le
 	 (1+\gamma\Delta t_i)
	 	 E \left[ | E _i\left[ Y_{t_{i+1}}-\widehat{\mathcal{U}}_{i+1}(X_{t_{i+1}})\right]|^2 \right]
		 +5(1+\gamma\Delta t_i)\frac{ L_2 ^2}{\gamma} \bigg\{
		 	C\rho(|\pi|)|\pi|+2\Delta t_i  E |Y_{t_i}-\widehat{\mathcal{V}}_{t_i}|^2
\nonumber\\
&\quad
+ E \left[\int_{t_i}^{t_{i+1}}\left( |Z_t-\bar{Z}_{t_i}|^2+ |\tilde{Z}_t-\bar{\tilde{Z}}_{t_i}|^2 \right)dt\right]
	+4\Big( E |Y_{t_{i+1}}-\widehat{\mathcal{U}}_{i+1}(X_{t_{i+1}})|^2\nonumber\\
&\quad- E\big[ | E _i[Y_{t_{i+1}}-\widehat{\mathcal{U}}_{i+1}(X_{t_{i+1}})]|^2\big] \Big)
	+4\Delta t_i  E \left[\int_{t_i}^{t_{i+1}}|F_t(e^{\mathcal{X}_t},Y_t,Z_t,\tilde{Z}_t)|^2dt\right]
	\bigg\}
\nonumber\\
&
\le  C\rho(|\pi|)|\pi|+(1+\gamma\Delta t_i) E |Y_{t_{i+1}}-\widehat{\mathcal{U}}_{i+1}(X_{t_{i+1}})|^2+C\Delta t_i E |Y_{t_i}-\widehat{\mathcal{V}}_{t_i}|^2 \nonumber\\
&\quad+C E \bigg[\int_{t_i}^{t_{i+1}} \Big(|Z_t-\bar{Z}_{t_i}|^2 + |\tilde{Z}_t-\bar{\tilde{Z}}_{t_i}|^2 \Big)dt\bigg]
+C\Delta t_i E \bigg[\int_{t_i}^{t_{i+1}}|F_t(e^{\mathcal{X}_t},Y_t,Z_t,\tilde{Z}_t)|^2dt\bigg],
\end{align}
which implies \eqref{4.12} when $|\pi|$ is sufficiently small.\\[6pt]
\textbf{Step 2.} We prove the estimate for the $Y$-component in \eqref{4.8}, i.e.,
 \begin{align}
\max _{i=0, \ldots, N-1}  E |Y_{t_{i}}-\widehat{\mathcal{U}}_{i}(X_{t_{i}})|^{2} &\leq C\rho(|\pi|)+C  E |G(\mathcal{X}_{T})-G(X_{T})|^{2}+C \varepsilon^{Z}(\pi)+C \varepsilon^{\tilde{Z}}(\pi)\nonumber\\
&\quad+C \sum_{i=0}^{N-1}(N \varepsilon_{i}^{\mathcal{N}, v}+\varepsilon_{i}^{\mathcal{N}, z}+\varepsilon_{i}^{\mathcal{N}, \tilde{z}}).\label{4.19}
\end{align}
Indeed,
using Young inequality of the form:
$$(a+b)^2 \ge  (1-|\pi|)a^2 + \left(1-\frac{1}{|\pi|}\right) b^2 \ge (1-|\pi|)a^2- \frac{1}{|\pi|}b^2,$$
we have
\begin{align}
 E |Y_{t_{i}}-\widehat{\mathcal{V}}_{t_{i}}|^{2}
 &= E |Y_{t_{i}}-\widehat{\mathcal{U}}_{i}(X_{t_{i}})+\widehat{\mathcal{U}}_{i}(X_{t_{i}})-\widehat{\mathcal{V}}_{t_{i}}|^{2}
 \nonumber\\
 &\geq(1-|\pi|)  E |Y_{t_{i}}-\widehat{\mathcal{U}}_{i}(X_{t_{i}})|^{2}-\frac{1}{|\pi|} E |\widehat{\mathcal{U}}_{i}(X_{t_{i}})-\widehat{\mathcal{V}}_{t_{i}}|^{2}. \label{4.13}
\end{align}
Plugging the above inequality into \eqref{4.12} and letting $|\pi|$ be small enough yield that
\begin{align}
 &E |Y_{t_{i}}-\widehat{\mathcal{U}}_{i}(X_{t_{i}})|^{2}
 \nonumber\\
 &
 \leq C\rho(|\pi|)|\pi|+(1+C|\pi|)  E |Y_{t_{i+1}}-\widehat{\mathcal{U}}_{i+1}(X_{t_{i+1}})|^{2}+C  E \bigg[\int_{t_{i}}^{t_{i+1}}\left(
 |Z_{t}-\bar{Z}_{t_{i}}|^{2} +|\tilde{Z}_{t}-\bar{\tilde{Z}}_{t_{i}}|^{2} \right) dt\bigg]
 \nonumber\\
&
 +C|\pi|  E \bigg[ \int_{t_{i}}^{t_{i+1}}|F_t(e^{\mathcal{X}_{t}}, Y_{t}, Z_{t}, \tilde{Z}_t)|^{2}dt\bigg]+CN  E |\widehat{\mathcal{V}}_{t_{i}}-\widehat{\mathcal{U}}_{i}(X_{t_{i}})|^{2}.
\end{align}
Recalling $Y_{t_N} = G(\mathcal{X}_T)$ and $\widehat{\mathcal{U}}_i(X_{t_N}) = G(X_T)$, and \eqref{F}, we may use the discrete Gronwall's inequality to reach the following estimate:
\begin{align}\label{4.14}
&\max _{i=0, \ldots, N-1}  E |Y_{t_{i}}-\widehat{\mathcal{U}}_{i}(X_{t_{i}})|^{2}
\notag\\
&\leq C \left\{
\rho(|\pi|) +|\pi|+ E |G(\mathcal{X}_{T})-G(X_{T})|^2+\varepsilon^{Z}(\pi)+\varepsilon^{\tilde{Z}}(\pi)+N\sum_{i=0}^{N-1}E|\widehat{\mathcal{U}}_i(X_{t_i})-\widehat{V}_{t_i}|^2 \right\},
\end{align}
which combined with Lemma \ref{lem-dis-v} gives \eqref{4.19}.

\noindent
\textbf{Step 3.} We prove the estimate for the $(Z,\tilde Z)$-component in \eqref{4.8},  i.e.,
  \begin{align*}
   &E \left[\sum_{i=0}^{N-1}\int_{t_{i}}^{t_{i+1}}\left( |Z_{t}-\widehat{\mathcal{Z}}_{i}(X_{t_{i}})|^{2} + |\tilde{Z}_{t}-\widehat{\tilde{\mathcal{Z}}}_{i}(X_{t_{i}})|^{2}\right) dt\right]
   \\
   &\leq C\left\{
   \varepsilon^Z(\pi)+\varepsilon^{\tilde{Z}}(\pi)+\rho(|\pi|)+|\pi|+ E |G(\mathcal{X}_T)-G(X_T)|^2
  +\sum_{i=0}^{N-1}(N\varepsilon^{\mathcal{N},v}_i+\varepsilon^{\mathcal{N},z}_i),+\varepsilon^{\mathcal{N},\tilde{z}}_i) \right\}.
  \end{align*}
From \eqref{4.10} and \eqref{4.11}, it follows that for any $i = 0,...,N-1$,
\begin{align*}
 &E [\int_{t_{i}}^{t_{i+1}}|Z_{t}-\overline{\widehat{Z}_{t_{i}}}|^{2}dt]
 \\
 &\leq  E\left[\int_{t_{i}}^{t_{i+1}}|Z_{t}-\bar{Z}_{t_{i}}|^{2}dt\right]
 	+2 \left( E |Y_{t_{i+1}}-\widehat{\mathcal{U}}_{i+1}(X_{t_{i+1}})|^2- E | E _i[Y_{t_{i+1}}-\widehat{\mathcal{U}}_{i+1}(X_{t_{i+1}})]|^2\right)\nonumber\\
&\quad
+2 |\pi| E \left[\int_{t_{i}}^{t_{i+1}}|F_t(e^{ \mathcal{X}_{t}}, Y_{t}, Z_{t}, \tilde{Z}_t)|^{2}dt\right].
\end{align*}
which, together with  \eqref{F}, gives
\begin{align}
 E \left[\sum_{i=0}^{N-1} \int_{t_{i}}^{t_{i+1}}|Z_{t}-\overline{\widehat{Z}_{t_i}}|^{2}dt\right]
 &\leq \varepsilon^{Z}(\pi)+2  E |G(\mathcal{X}_{T})-G(X_{T})|^{2}+2\sum_{i=0}^{N-1}\Big( E |Y_{t_{i}}-\widehat{\mathcal{U}}_{i}(X_{t_{i}})|^{2}\nonumber\\
&\quad -E | E _{i}[Y_{t_{i+1}}-\widehat{\mathcal{U}}_{i+1}(X_{t_{i+1}})]|^{2}\Big)+C|\pi|,\label{4.20}
\end{align}
where the indices are changed in the last summation. Analogously,
\begin{align}
 E \left[ \sum_{i=0}^{N-1} \int_{t_{i}}^{t_{i+1}}|\tilde{Z}_{t}-\overline{\widehat{\tilde{Z}}_{t_i}}|^{2}dt \right]
  &\leq \varepsilon^{\tilde{Z}}(\pi)+2  E |G(\mathcal{X}_{T})-G(X_{T})|^{2}
  	+2\sum_{i=0}^{N-1}\Big( E |Y_{t_{i}}-\widehat{\mathcal{U}}_{i}(X_{t_{i}})|^{2}\nonumber\\
&\quad -E | E _{i}[Y_{t_{i+1}}-\widehat{\mathcal{U}}_{i+1}(X_{t_{i+1}})]|^{2}\Big)+C|\pi|.
\end{align}
Notice that by \eqref{4.9} and \eqref{4.13} we have
\begin{align}
&2\left(  E |Y_{t_i}-\widehat{\mathcal{U}}_i(X_{t_i})|^2- E | E _i[Y_{t_{i+1}}-\widehat{\mathcal{U}}_{i+1}(X_{t_{i+1}})]|^2\right)
\nonumber\\
&
\le \frac{2}{1-|\pi|}\bigg\{(1+\gamma\Delta t_i) E \left[ | E _i[Y_{t_{i+1}}-\widehat{\mathcal{U}}_{i+1}(X_{t_{i+1}})]|^2 \right]
+5(1+\gamma\Delta t_i)\frac{ L_2 ^2}{\gamma} \Big( C\rho(|\pi|)|\pi|\nonumber\\
&\quad+2|\pi| E |Y_{t_i}-\widehat{\mathcal{V}}_{t_i}|^2+ E \Big[\int_{t_i}^{t_{i+1}}|Z_t-\overline{\widehat{Z}_{t_i}}|^2dt\Big]+ E \Big[\int_{t_i}^{t_{i+1}}|\tilde{Z}_t-\overline{\widehat{\tilde{Z}}_{t_i}}|^2dt \Big]\Big)  \bigg\}
\nonumber\\
&\quad+\frac{3}{|\pi|(1-|\pi|)} E |\widehat{\mathcal{U}}_i(X_{t_i} )-\widehat{\mathcal{V}}_{t_i}|^2.
\end{align}
Take $\gamma = 50 L_2 ^2$ so that $\frac{10 L_2 ^2}{\gamma}(1 + \gamma|\pi|)/(1-|\pi|) \le 1/4$ for $|\pi|$ small enough and notice that  $[(1+\gamma|\pi|)/(1-|\pi|)-1]= O(|\pi|)$.  This together with \eqref{F}, \eqref{4.18}, \eqref{4.12}, \eqref{4.19}, and  \eqref{4.20}, yields
{\small
\begin{align}
&\frac{1}{2} E \left[ \sum_{i=0}^{N-1}\int_{t_i}^{t_{i+1}}
\left(|Z_t-\overline{\widehat{Z}_{t_i}}|^2
+ |\tilde{Z}_t-\overline{\widehat{\tilde{Z}}_{t_i}}|^2 \right) dt \right]
\nonumber\\
&\le
\varepsilon^Z(\pi)+\varepsilon^{\tilde{Z}}(\pi)+C\max_{i=0,...,N} E |Y_{t_i}-\widehat{\mathcal{U}}_i(X_{t_i})|^2+C\rho(|\pi|)+C E |G(\mathcal{X}_T)-G(X_T)|^2\nonumber\\
&\quad+C|\pi|\sum_{i=0}^{N-1} E |Y_{t_i}-\widehat{\mathcal{V}}_{t_i}|^2+CN\sum_{i=0}^{N-1} E |\widehat{\mathcal{U}}_i(X_{t_i})-\widehat{\mathcal{V}}_{t_i}|^2 +C|\pi|
\nonumber\\
&\le
\varepsilon^Z(\pi)+\varepsilon^{\tilde{Z}}(\pi)+C \max_{i=0,...,N} E |Y_{t_i}-\widehat{\mathcal{U}}_i(X_{t_i})|^2+C\rho(|\pi|)+C|\pi|\nonumber\\
&\quad+
C|\pi|\sum_{i=0}^{N-1}\bigg\{
	C\rho(|\pi|)|\pi|+C E \bigg[ \int_{t_i}^{t_{i+1}} \Big(|Z_t-\bar{Z}_{t_i}|^2 + |\tilde{Z}_t-\bar{\tilde{Z}}_{t_i}|^2 \Big)dt\bigg]\nonumber\\
&\quad+(1+C|\pi|) E |Y_{t_{i+1}}-\widehat{\mathcal{U}}_{i+1}(X_{t_{i+1}})|^2
+C|\pi| E \bigg[\int_{t_i}^{t_{i+1}}|F(t,\mathcal{X}_t,Y_t,Z_t,\tilde{Z}_t)|^2dt\bigg] \bigg\}
\nonumber\\
&
\quad+CN\sum_{i=0}^{N-1} E |\widehat{\mathcal{U}}_i(X_{t_i})-\widehat{\mathcal{V}}_{t_i}|^2
\nonumber\\
&
\le C \bigg\{
\varepsilon^Z(\pi)+\varepsilon^{\tilde{Z}}(\pi)+\rho(|\pi|)+|\pi|+ E |G(\mathcal{X}_T)-G(X_T)|^2
+\sum_{i=0}^{N-1}(N\varepsilon^{\mathcal{N},v}_i+\varepsilon^{\mathcal{N},z}_i+\varepsilon^{\mathcal{N},\tilde{z}}_i)
\bigg\}.   \label{4.21}
\end{align}
}
 Finally, noticing the relations
\begin{align*}
 &E \left[\int_{t_{i}}^{t_{i+1}}\left|Z_{t}-\widehat{\mathcal{Z}}_{i}\left(X_{t_{i}}\right)\right|^{2} \mathrm{d} t\right] \leq 2  E \left[\int_{t_{i}}^{t_{i+1}}\left|Z_{t}-\overline{\widehat{Z}_{t_i}}\right|^{2} \mathrm{d} t\right]+2 \Delta t_{i}  E \left|\overline{\widehat{Z}_{t_i}}-\widehat{\mathcal{Z}}_{i}\left(X_{t_{i}}\right)\right|^{2},\\
 &E \left[\int_{t_{i}}^{t_{i+1}}\left|\tilde Z_{t}-\widehat{ \tilde {\mathcal{Z}}}_{i}\left(X_{t_{i}}\right)\right|^{2} \mathrm{d} t\right] \leq 2  E \left[\int_{t_{i}}^{t_{i+1}}\left|\tilde Z_{t}-\overline{\widehat{ {\tilde Z}}_{t_i}}\right|^{2} \mathrm{d} t\right]+2 \Delta t_{i}  E \left|\overline{\widehat{{\tilde Z}_{t_i}}}-\widehat{ \tilde{\mathcal{Z}}}_{i}\left(X_{t_{i}}\right)\right|^{2},
\end{align*}
and using \eqref{4.18}, \eqref{4.21}, we obtain by  summing over $i = 0,...,N-1,$ the desired error estimate for the $(Z,\tilde Z)$-component, completing the proof.
\end{proof}

Finally, we prove the claim in  Lemma \ref{lem-dis-v}.
\begin{proof}[Proof of Lemma \ref{lem-dis-v}]
Fix $i\in\{0,...,N-1\}.$ Using relation \eqref{4.7} in the expression of the expected quadratic loss function, and recalling the definitions of $\overline{\widehat{Z}_{t_i}}$ and $\overline{\widehat{\tilde{Z}}_{t_i}}$ as $L^2$-projection of $\widehat{Z}_t$ and $\widehat{\tilde{Z}}_t$, we have for all parameters $\theta$ of the neural networks $\mathcal{U}_i(.;\theta)$, $\mathcal{Z}_i(.;\theta)$, and $\tilde{\mathcal{Z}}_i(.;\theta)$,
\begin{align}
\hat{L}_{i}(\theta)=\tilde{L}_{i}(\theta)+ E \left[\int_{t_{i}}^{t_{i+1}} \left( \left| \widehat{Z}_{t}-\overline{\widehat{Z}_{t_{i}}} \right|^{2} + \left| \widehat{\tilde{Z}}_{t}-\overline{\widehat{\tilde{Z}}_{t_{i}}} \right|^{2}\right) \mathrm{d} t\right], \label{4.15}
\end{align}
with
\begin{align}
\tilde{L}_{i}(\theta)
:= &\, E \bigg[
	|\widehat{\mathcal{V}}_{t_{i}}-\mathcal{U}_{i}(X_{t_{i}}; \theta_i)
\nonumber\\
&
	\quad \quad+(F_{t_{i}}( e^{X_{t_{i}}}, \mathcal{U}_{i}(X_{t_{i}} ; \theta), \mathcal{Z}_{i}(X_{t_{i}}; \theta_i), \tilde{\mathcal{Z}}_{i}(X_{t_{i}}; \theta_i))
	-F_{t_i}( e^{X_{t_{i}}}, \widehat{\mathcal{V}}_{t_{i}}, \overline{\widehat{Z}_{t_{i}}}, \overline{\widehat{\tilde{Z}}_{t_{i}}}))\Delta t_{i}|^{2} \bigg]
\nonumber\\
&
	\quad\quad
		+\Delta t_{i}  E \left[ |\overline{\widehat{Z}_{t_{i}}}-\mathcal{Z}_{i}(X_{t_{i}};\theta_i)|^{2} \right]+\Delta t_{i}  E \left[ |\overline{\widehat{\tilde{Z}}_{t_{i}}}-\tilde{\mathcal{Z}}_{i}(X_{t_{i}};\theta_i)|^{2} \right].
\end{align}
By using Young inequality: $(a + b)^2 \le (1 + \gamma\Delta t_i)a^2 + (1 + \frac{1}{\gamma\Delta t_i})b^2$, together with the Lipschitz condition on $F$ in {\bfseries (H1)}, we  see that
\begin{align}
\tilde{L}_{i}(\theta)
&
 \leq(1+C \Delta t_{i})  E |\widehat{\mathcal{V}}_{t_{i}}-\mathcal{U}_{i}(X_{t_{i}}; \theta_i)|^{2}
\nonumber\\
&\quad
+C \Delta t_{i}  E \left[ |\overline{\widehat{Z}_{t_{i}}}-\mathcal{Z}_{i}(X_{t_{i}} ; \theta_i)|^{2}+  |\overline{\widehat{\tilde{Z}}_{t_{i}}}-\tilde{\mathcal{Z}}_{i}(X_{t_{i}} ; \theta_i)|^{2} \right]. \label{4.16}
\end{align}
On the other hand, using Young inequality in the form: $(a+b)^2 \ge (1-\gamma\Delta t_i)a^2 + (1-\frac{1}{\gamma\Delta t_i})b^2\ge (1-\gamma\Delta t_i)a^2-\frac{1}{\gamma\Delta t_i}b^2$, together with the Lipschitz condition on $F$, gives
\begin{align}
\tilde{L}_{i}(\theta) \geq&(1-\gamma \Delta t_{i}) E |\widehat{V}_{t_{i}}-\mathcal{U}_{i}(X_{t_{i}} ;\theta_i)|^{2}-\frac{3\Delta t_{i} L_2 ^{2}}{\gamma}\Big( E |\widehat{\mathcal{V}}_{t_{i}}-\mathcal{U}_{i}(X_{t_{i}} ;\theta_i)|^{2}+ E |\overline{\widehat{Z}_{t_i}}-\mathcal{Z}_{i}(X_{t_{i}} ;\theta_i)|^{2}\nonumber\\
&+ E |\overline{\widehat{\tilde{Z}}_{t_i}}-\tilde{\mathcal{Z}}_{i}(X_{t_{i}} ; \theta_i)|^{2}\Big)+\Delta t_{i}  E |\overline{\widehat{Z}_{t_i}}-\mathcal{Z}_{i}(X_{t_{i}} ; \theta_i)|^{2}+\Delta t_{i}  E |\overline{\widehat{\tilde{Z}}_{t_i}}-\tilde{\mathcal{Z}}_{i}(X_{t_{i}} ; \theta_i)|^{2}.
\end{align}
Choosing $\gamma=6 L_2 ^{2}$, this yields
\begin{equation}\label{4.17}
\tilde{L}_{i}(\theta) \geq(1-C \Delta t_{i})  E |\widehat{V}_{t_{i}}-\mathcal{U}_{i}(X_{t_{i}}; \theta_i)|^{2}+\frac{\Delta t_{i}}{2}  E\bigg[  |\overline{\widehat{Z}_{t_{i}}}-\mathcal{Z}_{i}(X_{t_{i}}; \theta_i)|^{2}+  |\overline{\widehat{\tilde{Z}}_{t_{i}}}-\tilde{\mathcal{Z}}_{i}(X_{t_{i}};\theta_i)|^{2} \bigg].
\end{equation}
For each $i\in\{0, \ldots, N-1\}$, take $\theta_{i}^{*}\in\arg\min _{\theta}\hat{L}_{i}(\theta)$ so that $\widehat{\mathcal{U}}_{i}=\mathcal{U}_{i}(\cdot ;\theta_{i}^{*})$, $\widehat{\mathcal{Z}}_{i}=\mathcal{Z}_{i}(\cdot; \theta_{i}^{*})$, and $\widehat{\mathcal{\tilde{Z}}}_{i}=\tilde{\mathcal{Z}}_{i}(\cdot; \theta_{i}^{*})$.
As the second term of the right hand side of \eqref{4.15} is independent of parameters $\theta_i$,  it also holds that $\theta_{i}^{*}\in\arg\min _{\theta}\tilde{L}_{i}(\theta)$. Combining \eqref{4.17} and \eqref{4.16} implies that for all $\theta$
\begin{align}
&(1-C \Delta t_{i})  E |\widehat{\mathcal{V}}_{t_{i}}-\widehat{\mathcal{U}}_{i}(X_{t_{i}})|^{2}
+\frac{\Delta t_{i}}{2}  E \left[|\overline{\widehat{Z}_{t_{i}}}-\widehat{\mathcal{Z}}_{i}(X_{t_{i}})|^2+  |\overline{\widehat{\tilde{Z}}_{t_{i}}}-\widehat{\tilde{\mathcal{Z}}}_{i}(X_{t_{i}})|^2\right]
\leq \tilde{L}_{i}(\theta_{i}^{*})
\leq
\tilde{L}_{i}(\theta)
 \nonumber\\
& \leq(1+C \Delta t_{i})  E |\widehat{\mathcal{V}}_{t_{i}}-\mathcal{U}_{i}(X_{t_{i}};\theta_i)|^{2}+C \Delta t_{i}  E \left[|\overline{\widehat{Z}_{t_{i}}}-\mathcal{Z}_{i}(X_{t_{i}};\theta_i)|^{2}+  |\overline{\widehat{\tilde{Z}}_{t_{i}}}-\tilde{\mathcal{Z}}_{i}(X_{t_{i}};\theta_i)|^{2}\right].
\end{align}
By \eqref{4.6}, letting $|\pi|$  be sufficiently small gives \eqref{4.18}.
\end{proof}

\bibliographystyle{alpha}

\begin{thebibliography}{BHM{\etalchar{+}}19}
	
	\bibitem[ALV07]{alos2007short}
	Elisa Al{\`o}s, Jorge~A Le{\'o}n, and Josep Vives.
	\newblock On the short-time behavior of the implied volatility for
	jump-diffusion models with stochastic volatility.
	\newblock {\em Finance and Stochastics}, 11(4):571--589, 2007.
	
	\bibitem[BD14]{BenderDokuchev-2014}
	Christian Bender and Nikolai Dokuchaev.
	\newblock A first-order {BSPDE} for swing option pricing.
	\newblock {\em Math. Finance}, 2014.
	\newblock DOI: 10.1111/mafi.12067.
	
	\bibitem[BDH{\etalchar{+}}03]{Hu_2002}
	P.~Briand, B.~Delyon, Y.~Hu, E.~Pardoux, and L.~Stoica.
	\newblock L$^p$ solutions of backward stochastic differential equations.
	\newblock {\em Stoch. Process. Appl.}, 108(4):604--618, 2003.
	
	\bibitem[BFG16]{bayer2016pricing}
	Christian Bayer, Peter Friz, and Jim Gatheral.
	\newblock Pricing under rough volatility.
	\newblock {\em Quantitative Finance}, 16(6):887--904, 2016.
	
	\bibitem[BFG{\etalchar{+}}19]{bayer2019regularity}
	Christian Bayer, Peter~K Friz, Paul Gassiat, Jorg Martin, and Benjamin Stemper.
	\newblock A regularity structure for rough volatility.
	\newblock {\em Mathematical Finance}, 2019.
	
	\bibitem[BHM{\etalchar{+}}19]{bayer2019deep}
	Christian Bayer, Blanka Horvath, Aitor Muguruza, Benjamin Stemper, and Mehdi
	Tomas.
	\newblock On deep calibration of (rough) stochastic volatility models.
	\newblock {\em arXiv preprint arXiv:1908.08806}, 2019.
	
	\bibitem[BL08]{buckdahnLi-2008-SDG-HJBI}
	Rainer Buckdahn and Juan Li.
	\newblock Stochastic differential games and viscosity solutions of
	{Hamilton-Jacobi-Bellman-Isaacs} equations.
	\newblock {\em SIAM J. Control Optim.}, 47(1):444--475, 2008.
	
	\bibitem[BTW18]{bayer2018pricing}
	Christian Bayer, Ra{\'u}l Tempone, and S{\"o}ren Wolfers.
	\newblock Pricing american options by exercise rate optimization.
	\newblock {\em arXiv preprint arXiv:1809.07300}, 2018.
	
	\bibitem[CCR12]{comte2012affine}
	Fabienne Comte, Laure Coutin, and Eric Renault.
	\newblock Affine fractional stochastic volatility models.
	\newblock {\em Annals of Finance}, 8(2-3):337--378, 2012.
	
	\bibitem[CH05]{cox2005local}
	Alexander~MG Cox and David~G Hobson.
	\newblock Local martingales, bubbles and option prices.
	\newblock {\em Finance and Stochastics}, 9(4):477--492, 2005.
	
	\bibitem[DPZ14]{da2014stochastic}
	Giuseppe Da~Prato and Jerzy Zabczyk.
	\newblock {\em Stochastic equations in infinite dimensions}.
	\newblock Cambridge university press, 2014.
	
	\bibitem[DQT11]{DuQiuTang10}
	Kai Du, Jinniao Qiu, and Shanjian Tang.
	\newblock $\textrm{L}^p$ theory for super-parabolic backward stochastic partial
	differential equations in the whole space.
	\newblock {\em Appl. Math. Optim.}, 65(2):175--219, 2011.
	
	\bibitem[EEFR18]{el2018microstructural}
	Omar El~Euch, Masaaki Fukasawa, and Mathieu Rosenbaum.
	\newblock The microstructural foundations of leverage effect and rough
	volatility.
	\newblock {\em Finance and Stochastics}, 22(2):241--280, 2018.
	
	\bibitem[EER19]{el2019characteristic}
	Omar El~Euch and Mathieu Rosenbaum.
	\newblock The characteristic function of rough {H}eston models.
	\newblock {\em Mathematical Finance}, 29(1):3--38, 2019.
	
	\bibitem[EKP{\etalchar{+}}97]{El_Karoui-reflec-1997}
	N.~E\textrm{l Karoui}, C.~Kapoudjian, E.~Paudoux, S.~Peng, and M.~C. Quenez.
	\newblock Reflected solutions of backward \textrm{SDE's}, and related obstacle
	problems for \textrm{PDE's}.
	\newblock {\em Ann. Probab.}, 25(2):702--737, 1997.
	
	\bibitem[EKTZ14]{ekren2014viscosity}
	Ibrahim Ekren, Christian Keller, Nizar Touzi, and Jianfeng Zhang.
	\newblock On viscosity solutions of path dependent {PDE}s.
	\newblock {\em The Annals of Probability}, 42(1):204--236, 2014.
	
	\bibitem[EPQ97]{Karoui_Peng_Quenez}
	N.~E\textrm{l Karoui}, S.~Peng, and M.~C. Quenez.
	\newblock Backward stochastic differential equations in finance.
	\newblock {\em Math. Finance}, 7(1):1--71, 1997.
	
	\bibitem[Fuk11]{fukasawa2011asymptotic}
	Masaaki Fukasawa.
	\newblock Asymptotic analysis for stochastic volatility: martingale expansion.
	\newblock {\em Finance and Stochastics}, 15(4):635--654, 2011.
	
	\bibitem[Gas18]{gassiat2018martingale}
	Paul Gassiat.
	\newblock On the martingale property in the rough {B}ergomi model, 2018.
	
	\bibitem[GJR18]{gatheral2018volatility}
	Jim Gatheral, Thibault Jaisson, and Mathieu Rosenbaum.
	\newblock Volatility is rough.
	\newblock {\em Quantitative Finance}, 18(6):933--949, 2018.
	
	\bibitem[GMZ20]{goudenege2020machine}
	Ludovic Gouden{\`e}ge, Andrea Molent, and Antonino Zanette.
	\newblock Machine learning for pricing {A}merican options in high-dimensional
	{M}arkovian and non-{M}arkovian models.
	\newblock {\em Quantitative Finance}, 20(4):573--591, 2020.
	
	\bibitem[HJW18]{han2018solving}
	Jiequn Han, Arnulf Jentzen, and E~Weinan.
	\newblock Solving high-dimensional partial differential equations using deep
	learning.
	\newblock {\em Proceedings of the National Academy of Sciences},
	115(34):8505--8510, 2018.
	
	\bibitem[HMY02]{Hu_Ma_Yong02}
	Y.~Hu, J.~Ma, and J.~Yong.
	\newblock On semi-linear degenerate backward stochastic partial differential
	equations.
	\newblock {\em Probab. Theory Relat. Fields}, 123:381--411, 2002.
	
	\bibitem[HP91]{HuPeng}
	Y.~Hu and S.~Peng.
	\newblock Adapted solution of a backward semilinear stochastic evolution
	equations.
	\newblock {\em Stoch. Anal. Appl.}, 9:445--459, 1991.
	
	\bibitem[HPW19]{hure2019some}
	C{\^o}me Hur{\'e}, Huy{\^e}n Pham, and Xavier Warin.
	\newblock Some machine learning schemes for high-dimensional nonlinear pdes.
	\newblock {\em arXiv preprint arXiv:1902.01599}, 2019.
	
	\bibitem[HSW89]{hornik1989multilayer}
	Kurt Hornik, Maxwell Stinchcombe, and Halbert White.
	\newblock Multilayer feedforward networks are universal approximators.
	\newblock {\em Neural networks}, 2(5):359--366, 1989.
	
	\bibitem[HSW90]{hornik1990universal}
	Kurt Hornik, Maxwell Stinchcombe, and Halbert White.
	\newblock Universal approximation of an unknown mapping and its derivatives
	using multilayer feedforward networks.
	\newblock {\em Neural networks}, 3(5):551--560, 1990.
	
	\bibitem[JLP19]{jaber2019affine}
	Eduardo~Abi Jaber, Martin Larsson, and Sergio Pulido.
	\newblock Affine {V}olterra processes.
	\newblock {\em The Annals of Applied Probability}, 29(5):3155--3200, 2019.
	
	\bibitem[JO19]{jacquier2019deep}
	Antoine~Jack Jacquier and Mugad Oumgari.
	\newblock Deep {PPDE}s for rough local stochastic volatility.
	\newblock {\em Available at SSRN 3400035}, 2019.
	
	\bibitem[Kry10]{Krylov_09}
	N.~V. Krylov.
	\newblock On the {It\^o}-{W}entzell formula for distribution-valued processes
	and related topics.
	\newblock {\em Probab. Theory Relat. Fields}, 150:295--319, 2010.
	
	\bibitem[Oks03]{oksendal2003stochastic}
	Bernt Oksendal.
	\newblock {\em Stochastic differential equations: an introduction with
		applications}.
	\newblock Springer, 2003.
	
	\bibitem[Pen92]{Peng_92}
	Shige Peng.
	\newblock Stochastic {H}amilton-{J}acobi-{B}ellman equations.
	\newblock {\em {SIAM} J. Control Optim.}, 30:284--304, 1992.
	
	\bibitem[PP90]{ParPeng_90}
	E.~Pardoux and S.~Peng.
	\newblock Adapted solution of a backward stochastic differential equation.
	\newblock {\em Syst. Control Lett.}, 14(1):55--61, 1990.
	
	\bibitem[Qiu17]{Qiu2014weak}
	Jinniao Qiu.
	\newblock Weak solution for a class of fully nonlinear stochastic
	hamilton--jacobi--bellman equations.
	\newblock {\em Stoch. Process. Appl.}, 127(6):1926--1959, 2017.
	
	\bibitem[Qiu18]{qiu2017viscosity}
	Jinniao Qiu.
	\newblock Viscosity solutions of stochastic {Hamilton--Jacobi--Bellman}
	equations.
	\newblock {\em {SIAM} J. Control Optim.}, 56(5):3708--3730, 2018.
	
	\bibitem[QW14]{QiuWei-RBSPDE-2013}
	Jinniao Qiu and Wenning Wei.
	\newblock On the quasi-linear reflected backward stochastic partial
	differential equations.
	\newblock {\em J. Funct. Anal.}, 267:3598--3656, 2014.
	
	\bibitem[VZ{\etalchar{+}}19]{viens2019martingale}
	Frederi Viens, Jianfeng Zhang, et~al.
	\newblock A martingale approach for fractional brownian motions and related
	path dependent pdes.
	\newblock {\em The Annals of Applied Probability}, 29(6):3489--3540, 2019.
	
	\bibitem[Zho92]{Zhou_92}
	Xun~Yu Zhou.
	\newblock A duality analysis on stochastic partial differential equations.
	\newblock {\em J. Funct. Anal.}, 103:275--293, 1992.
	
\end{thebibliography}
\newcommand{\etalchar}[1]{$^{#1}$}

\end{document}